\documentclass{article}

\PassOptionsToPackage{numbers}{natbib} 

\usepackage[preprint]{neurips_2025}

\usepackage{algorithm}
\usepackage{algorithmic}
\usepackage{multirow}
\usepackage{nicematrix}
\usepackage[utf8]{inputenc} 
\usepackage[T1]{fontenc}    
\usepackage{hyperref}       
\usepackage{url}            
\usepackage{booktabs}       
\usepackage{amsfonts}       
\usepackage{nicefrac}       
\usepackage{microtype}      
\usepackage{xcolor}         
\usepackage{xspace}
\usepackage{amsmath}
\usepackage{amssymb}
\usepackage{mathtools}
\usepackage{amsthm}
\usepackage{changes}
\definecolor{darkgreen}{rgb}{0.0, 0.5, 0.0}
\definecolor{darkred}{rgb}{0.5, 0.0, 0.0}
\definecolor{darkblue}{rgb}{0.0, 0.0, 0.5}
\newcommand{\ie}{{\emph{i.e.}},\xspace}
\newcommand{\eg}{{\emph{e.g.}}\xspace}
\newtheorem{lemma}{Lemma}
\newcommand{\ourmodel}{\textsc{NS-Pep}}

\title{\textsc{NS-Pep}: \textit{De novo} Peptide Design with Non-Standard Amino Acids}

\author{%
  Tao Guo\\
  Center of Excellence for Smart Health\\
  KAUST\\
  \texttt{tao.guo@kaust.edu.sa} \\
  \And
  Junbo Yin\\
  Center of Excellence for Smart Health\\
  KAUST\\
  \texttt{junbo.yin@kaust.edu.sa}\\
  \And
  Yu Wang\\
  Syneron Technology\\
  \texttt{yu.wang@synerontech.com}\\
  \And
  Xin Gao\thanks{Computer Science Program, Computer, Electrical and Mathematical Sciences and Engineering Division, King Abdullah University of Science and Technology (KAUST), Thuwal 23955-6900, Kingdom of Saudi Arabia. Center of Excellence on Generative AI, King Abdullah University of Science and Technology (KAUST), Thuwal 23955-6900, Kingdom of Saudi Arabia.}\\
  Center of Excellence for Smart Health\\
  KAUST\\
  \texttt{xin.gao@kaust.edu.sa}\\  
}

\begin{document}

\maketitle

\begin{abstract}
Peptide drugs incorporating non-standard amino acids (NSAAs) offer enhanced binding affinity and improved pharmacological properties. However, existing peptide design methods are limited to standard amino acids, leaving NSAA-aware design largely unexplored.
We introduce \textbf{\ourmodel}, a unified framework for co-designing peptide sequences and structures with NSAAs. The main challenge is that NSAAs are extremely underrepresented—even the most frequent one, SEP, accounts for less than 0.4\% of residues—resulting in a severe long-tailed distribution.
To improve generalization to rare amino acids, we propose \textbf{Residue Frequency-Guided Modification (RFGM)}, which mitigates over-penalization through frequency-aware logit calibration, supported by both theoretical and empirical analysis.
Furthermore, we identify that insufficient side-chain modeling limits geometric representation of NSAAs. To address this, we introduce \textbf{Progressive Side-chain Perception (PSP)} for coarse-to-fine torsion and location prediction, and \textbf{Interaction-Aware Weighting (IAW)} to emphasize pocket-proximal residues.
Moreover, \ourmodel~generalizes naturally to the peptide folding task with NSAAs, addressing a major limitation of current tools. Experiments show that \ourmodel~improves sequence recovery rate and binding affinity by 6.23\% and 5.12\%, respectively, and outperforms AlphaFold3 by 17.76\% in peptide folding success rate. The code will be released upon paper acceptance.
\end{abstract}

\section{Introduction} \label{Sec.Introduction}
Peptides, typically defined as short single-chain proteins \cite{pep_length}, play vital roles in various biological processes \cite{pep1} and have become increasingly important in drug development \cite{pep-drug,pep-drug2}. To enhance binding affinity and functional specificity, pharmaceutical chemists have expanded beyond standard amino acids by incorporating \textbf{N}on-\textbf{S}tandard \textbf{A}mino \textbf{A}cids (NSAAs) into peptide therapeutics \cite{NSAA1, NSAA2}. Traditionally, designing peptides with NSAAs has relied heavily on rational design guided by expert knowledge, which is time-consuming and not tractable to high-throughput workflows. Consequently, recent trend is to adopt AI-driven methods to accelerate and scale up the peptide design process.

Deep generative models \cite{DPM, DDPM, SGM, CFM} have demonstrated strong capabilities in the \textit{de novo} design of general proteins \cite{RFD, ProteinGenerator, CarbonNovo, FrameDiff, chroma, proteus}. However, peptide design poses unique challenges due to the higher conformational flexibility of peptides. To address this, recent methods \cite{DPB, ppflow, pepflow, pepglad} leverage receptor structures to stabilize peptide conformations and jointly design their sequences and structures.
Notably, these approaches mainly focus on standard amino acids and are not designed to accommodate NSAAs. When present, NSAAs are typically substituted with their standard counterparts, which can lead to suboptimal structure stability (\eg, unreal structures in Fig.~\ref{fig.SubstitutionNSAA}(\textbf{a})) and the loss of the unique physicochemical properties of NSAAs (\eg, worse affinity in Fig.~\ref{fig.SubstitutionNSAA}(\textbf{b})).
Therefore, \textit{how to develop a peptide sequence–structure co-design framework that supports NSAAs} remains an unsolved problem.

\begin{figure*}[t]
\begin{center}
\centerline{\includegraphics[width=\textwidth]{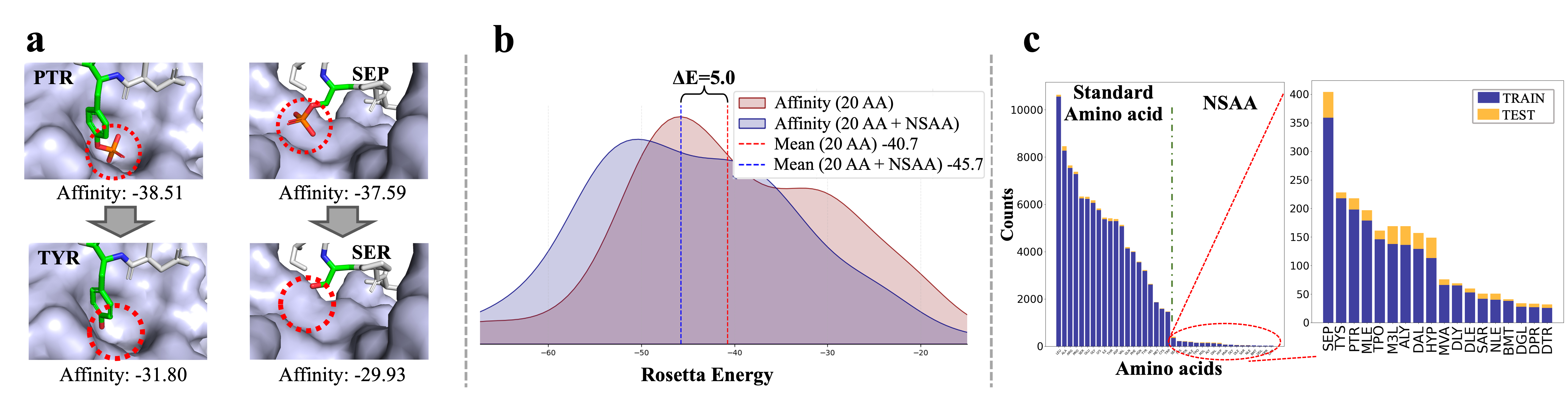}}
\vspace{-3mm}
\caption{\textbf{(a)} Naively substituting NSAAs (\eg, PTR and SEP) with standard ones (\eg, TYR and SER) results in suboptimal conformations. \textbf{(b)} Distributions of binding affinity before (\textcolor{darkblue}{blue}) and after (\textcolor{darkred}{red}) NSAA substitution. \textbf{(c)} Distribution of Standard and Non-standard Amino Acids.}
\label{fig.SubstitutionNSAA}
\end{center}
\vspace{-8mm}
\end{figure*}

Critically, we identify three fundamental challenges in peptide design with NSAAs.
\textbf{First}, the extreme scarcity of NSAA data presents a significant long-tailed learning challenge, as illustrated in Fig.~\ref{fig.SubstitutionNSAA}(\textbf{c}). For instance, the 18 most frequent NSAAs in our dataset (derived from Q-BioLip \cite{Q-BioLip} and PepBDB \cite{PepBDB}) collectively comprise around 2\% of the total.
Treating standard and non-standard amino acids equally during training will inevitably cause the model to be dominated by standard residues, resulting in representation collapse and poor performance on NSAAs.
\textbf{Second}, many NSAAs are derived from standard amino acids with subtle modifications to their side-chain functional groups. However, most existing generative models use only torsion angles to represent side chains \cite{DPB, pepflow, pepglad}, which struggle to capture the structural details critical for distinguishing NSAAs. For example, TYR (a standard amino acid) and DTR (an NSAA) exhibit similar torsion angle distributions (Fig.~\ref{Fig.torsion_dist}) but differ substantially in their side-chain atom composition. Capturing such differences accurately requires more fine-grained side-chain information, \eg, atomic level.
\textbf{Third}, the absence of a dedicated peptide folding model that supports both NSAA incorporation and binding site specification hinders progress in structure-based peptide binder analysis. For example, although AlphaFold3 \cite{af3} accommodates various NSAAs, its inability to specify target pockets often leads to off-target binding predictions. Conversely, model such as PepFlow \cite{pepflow} lacks support for NSAAs, limiting its utility in the design of practical peptide therapeutics.


To tackle the above challenges, we introduce \textbf{\ourmodel}, a unified flow-matching model that jointly supports the \textit{generation and structure prediction} of peptides with NSAAs. The core ideas of \ourmodel~are:
\textbf{(I)} A \textbf{R}esidue \textbf{F}requency-\textbf{G}uided \textbf{M}odification (RFGM) mechanism that mitigates the over-penalization of rare amino acids, which is achieved by calibrating the model’s output distribution based on residue occurrence frequencies.
\textbf{(II)} A \textbf{P}rogressive \textbf{S}ide-chain \textbf{P}erception (PSP) module that predicts torsion angles and atom-level offsets in a coarse-to-fine manner. To ensure compatibility between sequence and side chain, PSP decouples their co-generation and predicts side-chain conformations based on the generated residue type. Furthermore, to better capture critical side-chain interactions, PSP places greater emphasis on residues that are spatially proximal to the protein pocket.
\textbf{(III)} During training, \ourmodel~jointly predicts sequential and structural information from noised inputs, enabling it to capture the underlying mapping from sequence to structure. As a result, during inference, given a protein pocket and a peptide sequence, the model can accurately fold the peptide into the designated site. Notably, it generalizes naturally to peptides with NSAAs, addressing a key limitation of existing peptide folding models.

The main contributions of this work are as follows:
\begin{itemize}
 \item To the best of our knowledge, \textbf{\ourmodel} is the first generative framework that supports the co-design of peptides with non-standard amino acids (NSAAs). Moreover, \ourmodel~exhibits strong peptide folding capabilities for both standard and non-standard amino acids, advancing the development of peptide design models.
\item We propose a {Residue Frequency-Guided Modification (RFGM)} strategy to address the long-tailed issues posed by NSAAs. We provide both theoretical analysis and comprehensive empirical results to validate its effectiveness.
\item We introduce a Progressive Side-chain Perception module to more effectively leverage side-chain information for NSAA modeling. Additionally, an Interaction-Aware Weighting (IAW) module explicitly guides the model’s attention toward residues at the binding interface.
\end{itemize}
Experimental results demonstrate that \ourmodel~achieves improvements of 6.23\%  in AAR and 5.12\% in binding affinity for peptide design with NSAAs. On the pocket-specific peptide folding task, \ourmodel~outperforms AlphaFold3 by 20.4 \AA~in RMSD and achieves a 17.76\% higher success rate.


\section{Related Works}
\subsection{Generative Models for \textit{De novo} Protein Design}
 Recent years have witnessed the rapid advancement of deep generative models \cite{DPM, DDPM, SGM, CFM, RFM} in the field of \textit{de novo} protein design. As a pioneering approach, RFdiffusion \cite{RFD} generates high-quality protein backbones through structural diffusion, and subsequently employs ProteinMPNN \cite{ProteinMPNN} \cite{ProteinMPNN} to design the corresponding amino acid sequences.
After that, several approaches developed this backbone generation method in terms of theory and application \cite{FrameDiff, chroma, proteus}. Considering the consistency between structure and sequence, ProteinGenerator \cite{ProteinGenerator} and CarbonNovo \cite{CarbonNovo} expanded to protein sequence-structure co-design and achieved high consistency between these two modalities. 
Building on the progress in protein design, peptide generation methods have been developed to design short-chain binders targeting specific proteins \cite{DPB, ppflow, pepflow, pepglad, pephar}. While these approaches have shown promising results, they are restricted to the 20 standard amino acids and do not support NSAAs.
Moreover, these methods rarely consider detailed side-chain structures, which carry rich sequential and structural information. Although Pallatom \cite{pallatom} introduced a novel representation to encode 14 common atom types, it does not model interactions between the target and binder, and struggles to generalize to peptides with NSAAs due to the extremely limited data. By contrast, \ourmodel~ addresses the imbalanced issues posed by NSAAs while modeling side chains in a coarse-to-fine manner. It also incorporates target interface information to enable accurate sequence–structure co-design and folding for peptides.

\subsection{Long-tailed Learning}
Long-tailed learning aims to address the challenge of imbalanced data, where certain categories have remarkably fewer samples than others. This approach was initially applied in discriminative models. For instance, Menon \textit{et al.}~\cite{LogitAdj} introduced a statistical framework that included classification logit correction and loss modification. Building on this post-hoc strategy, Seesaw loss \cite{seesaw} was proposed to tackle long-tailed issues by carefully modifying gradients. Beyond post-processing methods, other approaches emphasized data sampling techniques, such as over-sampling tail classes \cite{oversampling, FSOD} or under-sampling dominant ones \cite{undersampling}. For example, TIMED-Design \cite{TIMED-Design-undersampling} improved model performance by discarding over-sampled amino acid types. Recently, long-tailed issues in generative models have received significant research interest. Several works \cite{CBDM, DiffROP, ContrastiveDiffusion} deployed contrastive learning between head and tail classes to generate better and more diverse samples for underrepresented categories. To tackle more fine-grained long-tailed learning challenges, Balancing Logit Variation \cite{BLV} adjusted the feature space by introducing class-frequency-related noise to perturb pixel-level logits. 
It is worth noting that most of these methods focus on imbalance problems in images. In contrast, \ourmodel~is an early effort to address the imbalance of amino acid types in peptide design with NSAAs, tackling the long-tailed challenge at the residue level, supported by both theoretical analysis and empirical validation.

\section{Methodology}
In Sec.\ref{sec.Pepflow}, we introduce the overall framework of \ourmodel~for co-designing peptide sequences and structures with NSAAs via flow matching. Next, we present RFGM to model the long-tailed distribution of NSAAs in Sec.\ref{sec.RFGM}. After that, we elaborate on our side-chain modeling in Sec.\ref{sec.SideChain}.

\begin{figure*}[ht]
\begin{center}
\centerline{\includegraphics[width=\textwidth]{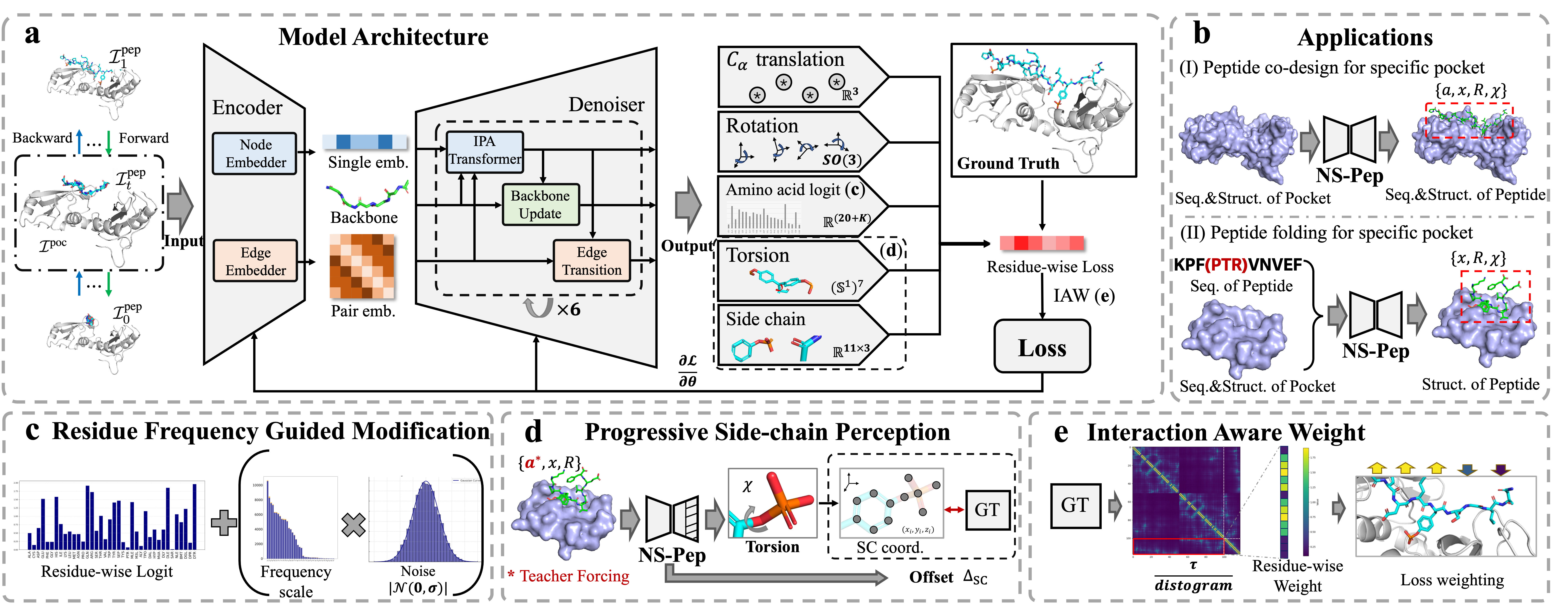}}
\vspace{-2mm}
\caption{
\textbf{Overview of \ourmodel.}
(\textbf{a}) The model takes a noised peptide and a protein pocket as input, and generates the peptide sequence and structure via flow matching. 
(\textbf{b}) \ourmodel~simultaneously supports NSAA-inclusion peptide co-design and folding with pocket conditioning. 
(\textbf{c}) RFGM facilitates the learning of long-tailed distributions. 
(\textbf{d}) PSP for coarse-to-fine side-chain structure modeling. 
(\textbf{e}) IAW guides the peptide generation toward functionally crucial hotspot residues.
}
\vspace{-8mm}
\label{fig.MainFig}
\end{center}
\end{figure*}

\subsection{Overview of \ourmodel} \label{sec.Pepflow}
Mathematically, a protein with $n$ residues is parameterized by the set of variables $\{ a^j, x^j, R^j, \chi^j \}_{j=1}^{n}$. The residue type $a^j \in \{1, \ldots, C\}$ includes the 20 standard amino acids and the $(C-20)$ most frequently used NSAAs considered in \ourmodel. For the $j$-th residue, $x^j \in \mathbb{R}^3$ and $R^j \in \mathbb{R}^{3 \times 3}$ represent the translation and rotation in the $N$–$C_\alpha$–$C$ backbone frame. The angle vector $\chi^j \in [0, 2\pi)^7$ consists of one $\psi^j$ angle for the backbone oxygen atom and six side-chain torsion angles. As shown in Fig.\ref{fig.MainFig}, the goal of \ourmodel~is to generate a peptide binder $\mathcal{I}^{\text{pep}} = \{ a^j, x^j, R^j, \chi^j \}_{j=1}^{n}$ of length $n$ including NSAAs, given a protein pocket $\mathcal{I}^{\text{poc}} = \{a^j, x^j, R^j, \chi^j \}_{j=1}^{m}$ of length $m$.

To model the variables of different manifolds, a Riemannian version of Conditional Flow Matching (CFM) \cite{CFM,RFM, pepflow} is applied. Given a variable $h\in\{a, x, R \}$ in a specific manifold, it leverages a time-dependent neural network $v_{t,\theta}(h_t, t): \mathbb{R}^d \times [0,1] \rightarrow \mathbb{R}^d$ to approximate a conditional vector field $u_{t}(h_0, h_1, t) : \mathbb{R}^d \times \mathbb{R}^d \times [0,1] \rightarrow \mathbb{R}^d$, which is denoted as:
\begin{equation}
   \mathcal{L}_{\text{CFM}}(\theta)=\mathbb{E}_{t,p_{1}(h_1),p_{0}(h_0)}||v_{t,\theta}(h_t, t) - u_{t}(h_0, h_1, t)||^2
\end{equation}
where $u_{t}(h_0, h_1, t)$ refers to $\frac{d}{dt}[\text{exp}_{h_{0}}(t \text{log}_{h_0}(h_1))]$ at the manifold. After model training, $v_{t,\theta}$ can generate samples starting from $h_0$ by solving the ordinary differential equation $\frac{dh_t}{dt}=v_{t,\theta}(h_t, t)$.

In this paper, we propose a peptide generation model, \ourmodel, which first generates $\{ a^j, x^j, R^j \}_{j=1}^{n}$ via flow matching and then predicts the side-chain information. The total loss is as follows:
\begin{equation}
    \mathcal{L} = \mathbb{E}_{t}[\sum^{n}_{j=1} w_{j}(\sum_{k=\{x,R\}} \lambda_{k}  \mathcal{L}_{\text{CFM},k}^{j} + \lambda_{a} \mathcal{L}_{\text{CFM},a}^{j} + \lambda_{\text{SC}} \mathcal{L}_{\text{PSP}}^{j})]
    \label{eq.TotalLoss}
\end{equation}
where $\mathcal{L}_{\text{CFM},a}^{j}$ is the flow-matching loss for residue type generation using modified logits to address long-tailed distribution issues (Sec.~\ref{sec.RFGM}); $\mathcal{L}_{\text{PSP}}^{j}$ denotes the progressive side-chain perception loss for side-chain prediction (Sec.~\ref{sec.SideChain}); $w_j$ is the proposed interaction-aware weight (Sec.~\ref{sec.SideChain}) applied to residue-wise losses; and $\lambda_{x},\lambda_{R},\lambda_{a},\lambda_{\text{SC}}$ denotes hyperparameters.

\subsection{Residue Frequency Guided Modification} \label{sec.RFGM}
Since NSAAs account for only a small fraction of all amino acids (\eg, only $\sim$2\%), it is extremely challenging for a generative model to accurately capture information from these long-tailed classes. As a result, the model prediction is often dominated by \textit{head} classes (standard amino acids), while \textit{tail} classes (NSAAs) are underrepresented and insufficiently learned during training.
We observe that a core issue lies in the training process: the logits of tail classes receive cumulatively large penalties. To address it, we propose \textbf{R}esidue \textbf{F}requency-\textbf{G}uided \textbf{M}odification (RFGM), which calibrates the gradient penalties by injecting frequency-weighted perturbations into the residue-type logits (Fig.~\ref{fig.MainFig}(\textbf{c})).

Formally, for the $j$-th position in the peptide sequence, the residue-type flow matching loss $\mathcal{L}_{\text{CFM},a}^{j}$ in Eq.~\ref{eq.TotalLoss} induces to a standard cross-entropy loss:
\begin{equation}
\mathcal{L}_{\text{CE}}(\mathbf{z}, \mathbf{y}) = -\sum^{C}_{i=1}y_{i}\text{log}(p_{i}), \ \text{with} \ p_{i}=\frac{e^{z_{i}}}{\sum^{C}_{j=1}e^{z_{j}}} \in (0,1),
\end{equation} 
where $p_i$ is the predicted probability of the $i$-th amino acid type, and $y_i \in \{0, 1\}$ is the corresponding one-hot label. These probabilities are computed from the logits $\mathbf{z} = [z_1, \ldots, z_C]$ via softmax, where $z_i$ is the logit for class $i$ and $C$ is the total number of amino acid types. For notational simplicity, we write $\mathcal{L}_{\text{CE}}(\mathbf{z})$ in place of $\mathcal{L}_{\text{CE}}(\mathbf{z}, \mathbf{y})$ throughout the following discussion. Therefore, one can derive the gradient of $\mathcal{L}_{\text{CE}}(\mathbf{z})$ with respect to $z_i$:
\begin{equation}
\frac{\partial \mathcal{L}_{\text{CE}}(\textbf{z})}{\partial z_{i}} = p_{i} - y_i
\end{equation}

During training on samples with head-class labels, tail classes appear as negative labels, meaning that the their class are not the ground truth. In such cases, the gradient of the cross-entropy loss with respect to the logit $z_i$ is
$\frac{\partial \mathcal{L}_{\text{CE}}(\mathbf{z})}{\partial z_i} = p_i \in (0,1)$.
Under gradient descent, this results in an update of $-p_i$, which is repeatedly applied to the output neurons corresponding to these tail classes. This leads to cumulatively large gradient penalties and tends to over-suppress the predicted logits of rare classes, making them difficult to recall during inference.

To mitigate the above issue, we propose a Residue Frequency-Guided Modification (RFGM) strategy, inspired by the long-tailed learning works~\cite{seesaw, BLV}. It aims to reduce the gradient penalty on tail classes by lowering their predicted probabilities $p_i$ \textbf{only} for loss calculation. Specifically, RFGM introduces class-frequency-aware perturbations to all elements of logits $\mathbf{z}$:
\begin{equation}
z_{i}' = z_{i} + \frac{\text{max} \ \mathbf{v}}{v_{i}} |\delta(\sigma)|, \ \text{with} \ v_i = \text{log} \frac{\sum^{C}_{j=1}n_{j}}{n_{i}},
\label{eq.RFGM}
\end{equation}
where $\mathbf{v}=[v_1, ..., v_{C}]$ are frequency-related weights, $n_i$ is the number of instances of class $i$, $\delta(\sigma)$ is a Gaussian variable with zero mean and standard deviation $\sigma$.
Given that the weight $v_i$ is inversely proportional to its category proportion, the scale of RFGM, \ie $\frac{\text{max} \ \mathbf{v}}{v_{i}}$, \textbf{increases} with class frequency. Although the scaling behavior introduced by RFGM may appear counter-intuitive, we demonstrate that it inclines to reduce the gradient magnitude for tail classes after logit modification, \ie $\frac{\partial \mathcal{L}_{\text{CE}}(\mathbf{z}')}{\partial z_i} \leq \frac{\partial \mathcal{L}_{\text{CE}}(\mathbf{z})}{\partial z_i}$, from both probabilistic and expectation perspectives (please see Appendix~\ref{sec.AnalysisRFGM}). 

As a result, when a tail class is treated as a negative label, RFGM mitigates its gradient penalty, preventing its corresponding logit from being overly suppressed. Conversely, head-class logits (when treated as negative) receive relatively stronger penalties, encouraging the model to better distinguish among frequent classes. Furthermore, the introduction of uncertainty via $|\delta(\sigma)|$ into the logits acts analogously to noise-denoising strategies in self-supervised learning \cite{ESM, self-supervised}, thereby enhancing representation learning in sequence prediction tasks (Tab.~\ref{Fig.ablation_noise}).

\subsection{Side-Chain Learning of \ourmodel} \label{sec.SideChain}

Protein side chains encode rich sequential and structural information, which is especially crucial for NSAAs due to their data scarcity. Sufficient side-chain learning enables the model to capture the unique characteristics of NSAAs and distinguish them from standard amino acids. In this section, we introduce two strategies to leverage side-chain information:
(I) \textbf{P}rogressive \textbf{S}ide-chain \textbf{P}erception (PSP); and
(II) \textbf{I}nteraction-\textbf{A}ware loss \textbf{W}eighting (IAW).

\textbf{Progressive Side-chain Perception}. Prior work such as PepFlow~\cite{pepflow} jointly generates amino acid sequences and their corresponding side chains. However, this may introduce inconsistencies and confuse the model. For instance, before the residue type is determined, the associated rigid groups—represented by torsion angles—remain ambiguous. Moreover, during generation, the continuous evolution of the sequence can lead to unstable or incoherent side-chain predictions.
Beyond this, torsion angles alone provide a coarse representation of side chains: they only describe rotations around rotatable bonds, while ignoring the detailed atomic geometry. Our statistical analysis further confirms this limitation—several amino acids with substantially different side-chain structures share nearly identical torsion angle distributions (\ie Fig.~\ref{Fig.torsion_dist}), making it difficult for the model to distinguish them. These observations highlight the need to both decouple sequence and side-chain modeling during generation, and to integrate a finer-grained, atom-level representation for side chains, as shown in Fig.~\ref{fig.MainFig}(\textbf{d}).

To decouple sequence and side-chain generation, we first generate residue identities via flow matching, and predict the corresponding side-chain conformations conditioned on the assigned residues in a progressive manner. In other words, we decompose the peptide co-design task into two stages: a generation task for $p(\{a^j,x^j,R^j\}^{n}_{j=1}|\mathcal{I}^{\text{poc}})$ and a side-chain prediction task, \ie $p(\{\chi^j, \Delta_{\text{SC}}^j \}^{n}_{j=1}|\mathcal{I}^{\text{poc}}, \{a^j,x^j,R^j\}^{n}_{j=1})$. Note that $\Delta_{\text{SC}}^j$ denotes the difference in side-chain conformation between the constructed structure and the reference structure, measured in the local coordinate frame. Inspired by UniGEM \cite{unigem} from the molecular modeling domain, we train the generator and side-chain predictor jointly, sharing part of their parameters to facilitate knowledge transfer between the two tasks. Specifically, during training, the generator is tasked with recovering the noised sequence and backbone structure:
\begin{equation}
    \{\hat{a}^j,\hat{x}^j,\hat{R}^j\}^{n}_{j=1} = \mathcal{F}_\theta(\mathcal{I}^{\text{poc}}, \{a^j_t,x^j_t,R^j_t\}^{n}_{j=1})
\end{equation}
where $\mathcal{F}_\theta$ denotes the proposed \ourmodel~model, and $a^j_t$ and $\hat{a}^j$ represent the noised input sequence and the predicted (denoised) sequence for the $j$-th residue, respectively. Other notations follow the same convention. 

To ensure structural stability, the side-chain predictor is only activated when $t > 0.75$, where the backbone is sufficiently refined. It then estimates torsion angles and side-chain atom offsets:
\begin{equation}
    \{\hat{\chi}^j, \Delta_{\text{SC}}^j \}^{n}_{j=1} = \mathcal{F}_\theta(\mathcal{I}^{\text{poc}}, \{a^j,\hat{x}^j,\hat{R}^j\}^{n}_{j=1})
    \label{eq.chi_pred}
\end{equation}
where $a^j$ refers to the ground-truth residue identity. The purpose is to align the labels of $a^j$ and $\chi^j$ and stabilize the training. Compared to the generation task, side-chain prediction is also straightforward.

To bridge the predicted torsion angles and the full side-chain geometry, we reconstruct the all-atom conformation using $\hat{\chi}$ as follows:
\begin{equation}
    \hat{X}^j = \mathcal{M}(a^j, \hat{x}^j, \hat{R}^j, \hat{\chi}^j)
\end{equation}
where $\mathcal{M}$ is the mapping operation that uses idealized atomic coordinates to assemble the all-atom peptide structure, following the conventions established by AlphaFold2 \cite{AF2}. Additionally, an all-atom structure is constructed to serve as the reference, \ie $X^j = \mathcal{M}(a^j, \hat{x}^j, \hat{R}^j, \chi^j)$. Therefore, the PSP loss in Eq.~\ref{eq.TotalLoss} is formulated as:
\begin{equation}
    \mathcal{L}^j_{\text{PSP}} = 
    ||\text{wrap}(\hat{\chi}^j) -  \text{wrap}(\chi^j)||^2 + \lambda_{\text{aa}} ||\Delta_{\text{SC}}^j -  \mathbb{T}^{-1}_{\{ \hat{x}^j, \hat{R}^j \}}(X^j - \hat{X}^j)||^2
    \label{eq.psp}
\end{equation}
where $\text{wrap}(u) = (u + \pi) \bmod (2\pi) - \pi$ maps angles to the interval $[-\pi, \pi)$, and $\lambda_{\text{aa}}$ is a weighting hyperparameter; $\mathbb{T}^{-1}_{\{ \hat{x}^j, \hat{R}^j \}}(\cdot)$ is a rigid transformation from global to local coordinates. The second term focuses exclusively on the side-chain geometry differences between the predicted and reference structures, as the backbone components $\{\hat{x}^j, \hat{R}^j\}$ are shared in both cases. During inference, we first construct the all-atom structure using the predicted torsion angles $\hat{X}$, and then apply the predicted offsets $\mathbb{T}_{\{ \hat{x}^j, \hat{R}^j \}}(\Delta_{\text{SC}}^j)$ to calibrate the final structures.

\textbf{Interaction-Aware Weighting}. While peptide structures are typically flexible, a subset of residues near the target pocket tends to adopt conserved conformations to facilitate binding. To emphasize these key residues during generation, we apply \textbf{I}nteraction-\textbf{A}ware \textbf{W}eighting (IAW), assigning weights according to each residue's distance from the binding interface, as shown in Fig.~\ref{fig.MainFig}(\textbf{e}).

Specifically, for an $n$-residue peptide bound to an $m$-residue protein, we compute a distance matrix $D \in \mathbb{R}^{n \times m}$, including side-chain interactions. For the $j$-th residue of peptide and $k$-th residue of protein, the entry $D_{j,k}$ denotes the Euclidean distance between the closest pair of atoms from these two residues. 
Thereby, the weight for the $j$-th peptide residue is defined as:
\begin{equation}
    w_{j} = \frac{\tau}{\text{min}_{k=1}^{m}(D_{j, k})}, \ j = \{ 1, \cdots, n\}
    \label{eq.IAW}
\end{equation}
where $\tau$ is a predefined distance threshold. Intuitively, residues closer to the protein receive higher weights, prompting the model to focus on key interactions during training (Eq.~\ref{eq.TotalLoss}). Since $D$ is derived from the ground-truth structure, these weights can be computed before training and reused throughout.

\section{Experiment}
\label{sec.exp}
In this section, we begin by benchmarking the peptide-generation performance of \ourmodel~against state-of-the-art peptide (protein) generative models (Sec.\ref{Exp.Comparison}). We then evaluate its ability to fold peptides into designated protein pockets (Sec.\ref{Exp.Folding}). A representative case study follows, illustrating the model’s decision-making process (Sec.\ref{Exp.Explain}). Finally, an ablation study quantifies the contribution of each proposed component (Sec.\ref{Exp.Ablation}). Implementation details and pseudo codes are provided in Appendix~\ref{sec.Implementation}, and further results in Appendix~\ref{sec.AdditionalResult}.

\textbf{Dataset.} In this work, the benchmark dataset derives from Q-BioLip \cite{Q-BioLip} and PepBDB \cite{PepBDB}, comprising 10,348 protein-peptide complexes in total. Peptide lengths range from 3 to 25 residues. We split the dataset into train: validation: test = 8,828: 1,065: 219 according to sequential cluster \cite{mmseqs2} and NSAA proportions. It should be mentioned that the test set with 219 complexes is call \textbf{General} test set.
We also extract a subset from the General test set in which every complex contains NSAAs, referred to as the \textbf{NSAA} test set. This subset is used to magnify the NSAA's impact on affinity. More details about data splitting are shown in Appendix~\ref{sec.DataSplit}.

\subsection{General \textit{De novo} Peptide Design} \label{Exp.Comparison}
This task evaluates eight peptide(protein) generation methods comprehensively with multiple metrics.

\begin{table*}[t]
\centering
\caption{Comparison for peptide generation on the \textbf{General} test set}
\resizebox{\textwidth}{!}{%
\begin{tabular}{lcccccccc}
\toprule
\multirow{2}{*}{Method} & \multicolumn{5}{c}{Recovery} & \multicolumn{1}{c}{Energy} & \multicolumn{2}{c}{Design} \\
\cmidrule(r){2-6} \cmidrule(r){7-7} \cmidrule(r){8-9}
 & AAR \% $\uparrow$ & AAR(S) \% $\uparrow$ & RMSD \AA\ $\downarrow$ & SSR \% $\uparrow$ & BSR \% $\uparrow$ & AFF \% $\uparrow$ & scRMSD \AA\ $\downarrow$ & Diversity $\uparrow$ \\
\midrule
RFdiffusion \cite{RFD} & 8.66 & 8.69 & 17.57 & 71.75 & 6.52 & 12.06 & 19.64 & 0.41 \\
ProteinGenerator \cite{ProteinGenerator} & 7.08 & 7.12 & 17.26 & 75.30 & 4.28 & 8.91 & 23.02 & 0.39 \\
PPFLOW \cite{ppflow} & 8.63 & 8.87 & 9.45 & 0.11 & 75.14 & 6.53 & 15.77 & \underline{0.66} \\
DiffPepBuilder \cite{DPB} & 6.71 & 6.77 & 8.36 & 75.49 & \textbf{85.86} & \underline{12.11} & 14.65 & 0.51 \\
PepGLAD \cite{pepglad} & 6.82 & 6.90 & 7.78 & 78.72 & 81.48 & 10.36 & 13.43 & \textbf{0.72} \\
PepFlow \cite{pepflow} & \underline{16.30} & \underline{16.76} & 4.98 & \textbf{85.80} & 81.50 & 11.60 & 12.32 & 0.50 \\
PepFlow* & 15.88 & 15.88 & \underline{4.94} & 83.40 & 84.36 & 11.39 & \underline{11.67} & 0.51 \\
\ourmodel(ours) & \textbf{22.53} {\scriptsize  \textcolor{darkgreen}{\textbf{+6.23}}} & \textbf{22.35} {\scriptsize  \textcolor{darkgreen}{\textbf{+5.59}}} & \textbf{4.92} {\scriptsize  \textcolor{darkgreen}{\textbf{-0.02}}} & \underline{85.00} {\scriptsize  \textcolor{darkred}{\textbf{-0.80}}} & \underline{81.62} {\scriptsize  \textcolor{darkred}{\textbf{-4.24}}} & \textbf{17.23} {\scriptsize  \textcolor{darkgreen}{\textbf{+5.12}}} & \textbf{11.50} {\scriptsize  \textcolor{darkgreen}{\textbf{-0.17}}} & 0.49 {\scriptsize  \textcolor{darkred}{\textbf{-0.23}}} \\
\bottomrule
\multicolumn{8}{l}{\small \textbf{Note:} (i) The best values are \textbf{bold}; (ii) The second best values are \underline{underlined};}\\
\multicolumn{8}{l}{\small \ \ \ \ \ \ \ \ \ \ (iii) \textcolor{darkgreen}{Green} and \textcolor{darkred}{red} font colors indicate that our performances are better or worse than the other best ones by certain values.}\\
\end{tabular}%
}
\vspace{-8mm}
\label{tab.comparison}
\end{table*}

\textbf{Baselines.} Seven state-of-the-art generative models are chosen as baselines, including RFdiffusion \cite{RFD}, ProteinGenerator \cite{ProteinGenerator}, PPFLOW \cite{ppflow}, DiffPepBuilder \cite{DPB}, PepGLAD \cite{pepglad}, PepFlow \cite{pepflow} and PepFlow* (expanded to NSAAs). More implementation details of baselines are shown in Appendix~\ref{sec.BaselineDetail}.

\textbf{Metrics.} 
Generated peptides are evaluated from three aspects with a total of eight metrics. (i) Recovery. Amino Acid Recovery (\textbf{AAR}) calculates the ratio of the same sequential identity between the ground truth and designed peptides, including NSAAs and standard ones. \textbf{AAR(S)} first substitutes all NSAAs of the ground truth peptides with their standard counterparts, than calculates the recovery. The \textbf{RMSD} is a structural recovery metric that aligns the complex by \textit{pockets} first and measures the RMSD between $C_{\alpha}$ atoms of peptides. Secondary-Structure Recovery (\textbf{SSR}) and Binding-Site Recovery (\textbf{BSR}) denotes the ratios of shared secondary structure and binding sites, respectively. (ii) Energy. The \textbf{AFF} metric computes the percentage of generated samples whose binding affinity, calculated by Rosetta \cite{rosetta}, is lower than that of the native binder. (iii) Design. The \textbf{scRMSD} refers to self-consistency RMSD that compares the difference between generated and predicted structures in terms of $C_{\alpha}$ atoms \cite{FrameDiff}. Note that, given the generated peptide sequence and the sequence of the protein target, we use Chai-1 \cite{Chai-1} to predict the structure of the peptide-protein complex. Besides, the \textbf{Diversity} metric is the average of one minus pairwise TM-score \cite{tm-score} among the generated samples, indicating the divergence of peptide structure. More descriptions of evaluation metrics can be found in Appendix~\ref{sec.AppendixExp}.

\textbf{Results.} 
As shown in Tab.~\ref{tab.comparison}, \ourmodel~significantly enhances AAR and AFF by 6.23\% and 5.12\%, respectively, compared to the second best models. Note that PepFlow* falls short of the original PepFlow in AAR, suggesting that a naive extension to NSAAs struggles with the long-tailed residue distribution.
For structural recovery, a high BSR is vital for achieving a low RMSD, as RMSD also depends on the alignment of the complex with respect to the protein pocket. There, methods such as RFdiffusion and ProteinGenerator achieve worse RMSD because of lower BSRs. Although \ourmodel\ does not attain the highest BSR, it excels in AFF, indicating that it effectively captures key residues.
\ourmodel, on the other hand, achieves the lowest scRMSD, indicating the highest sequence–structure consistency among all methods. As a result, \ourmodel, which attains the highest AAR, also demonstrates the lowest RMSD. Multiple cases are visualized in Fig.~\ref{Fig.exp}(\textbf{a}) and Fig.~\ref{Fig.App.Visualization_v1}.

\subsection{Peptide Folding} \label{Exp.Folding}
This experiment evaluates the peptide folding performance when targeting specific protein pockets and incorporating NSAAs. Apart from RMSD, BSR, and Diversity, we introduce two metrics: Success and AFF(Success). \textbf{Success} denotes the proportion of predicted structures with an RMSD below 2 \AA~ relative to the ground truth. \textbf{AFF(Success)} measures the percentage of successful samples whose predicted binding affinity is lower than that of the native binder. We also report two binary indicators: \textbf{NSAA} and \textbf{Binding Site}, denoting whether the folding model supports NSAA inclusion and binding site specification, respectively.

\begin{table}[ht]
\centering
\caption{Comparison for peptide folding on the \textbf{General} test set}
\label{tab.folding}
\resizebox{0.9\textwidth}{!}{%
\begin{tabular}{l|cc|ccccc}
\toprule
Method & NSAA & Binding Site & RMSD \AA $\downarrow$ & BSR \% $\uparrow$ & Success \% $\uparrow$ & AFF(Success) \% $\uparrow$ & Diversity $\uparrow$ \\
\midrule
ESMFold \cite{ESMFold} &  &  & 18.89 & 36.14 & 12.87 & 0.00 & 0.00 \\
AlphaFold3 \cite{af3} & \checkmark &   & 25.30 & 30.04 & 10.49 & 1.77 & 0.43 \\
PepFlow \cite{pepflow} &  & \checkmark & \underline{4.96} & \textbf{81.63} & \underline{25.52} & \underline{21.46} & \textbf{0.51}  \\
\ourmodel  & \checkmark & \checkmark &  \textbf{4.90} {\scriptsize  \textcolor{darkgreen}{\textbf{-0.06}}}  & \underline{81.61} {\scriptsize  \textcolor{darkred}{\textbf{-0.02}}} & \textbf{28.25} {\scriptsize  \textcolor{darkgreen}{\textbf{+2.73}}} & \textbf{25.45} {\scriptsize  \textcolor{darkgreen}{\textbf{+3.99}}} & \underline{0.50} {\scriptsize  \textcolor{darkred}{\textbf{-0.01}}}  \\
\bottomrule
\end{tabular}%
}
\end{table}

\textbf{Results.} As shown in Table~\ref{tab.folding}, both ESMFold \cite{ESMFold} and AlphaFold3 \cite{af3} exhibit poor RMSD performance, which can be attributed to their inability to specify binding pockets, as reflected by their low BSR scores. The case in Fig.\ref{Fig.exp}(\textbf{b}) illustrates that AlphaFold3 fails to fold the peptide into the designated pocket.

Compared to PepFlow, \ourmodel~achieves notably higher performance in both Success and AFF(Success). We hypothesize that AFF(Success) reflects a model’s ability to effectively incorporate NSAAs: when achieving similarly low ($C_{\alpha}$) RMSD, peptides containing native NSAAs can bind more tightly to the pocket than those using substituted residues, due to their ability to better mimic real side-chain interactions.

Moreover, when stratified by peptide length (Fig.\ref{Fig.exp}(\textbf{b})), \ourmodel~consistently outperforms PepFlow in both short and medium-length scenarios. Taken together, these results demonstrate that \ourmodel~fills an important gap, providing a peptide folding model that supports both NSAA incorporation and binding site specification.

\subsection{Case Study of \ourmodel} \label{Exp.Explain}

To illustrate how \ourmodel~determines NSAA placement, we examine case 3RM0 (PDB ID) in Fig.\ref{Fig.explain}, focusing on the $9^{\text{th}}$ residue of the peptide, which interacts with the $81^{\text{st}}$ residue in the receptor:
(i) At generation step $t = 0.5$, the IPA attention map reveals that the $81^{\text{st}}$ residue of receptor attends most strongly to the $9^{\text{th}}$ peptide residue, demonstrating \ourmodel’s ability to identify key interaction sites.
(ii) At this position, the top-4 generated residues are TYS (ground truth), PTR, TRP, and GLU—all of which share relevant physicochemical properties with the ground truth: PTR closely resembles TYS in both aromaticity and charge, TRP offers an aromatic ring, and GLU introduces a negative charge. These results highlight \ourmodel’s capacity to model both chemical and structural context effectively. Additional interpretability examples are provided in Appendix~\ref{sec.App.Explain}.

\begin{figure*}[ht]
\vspace{-2mm}
\begin{center}
\centerline{\includegraphics[width=\textwidth]{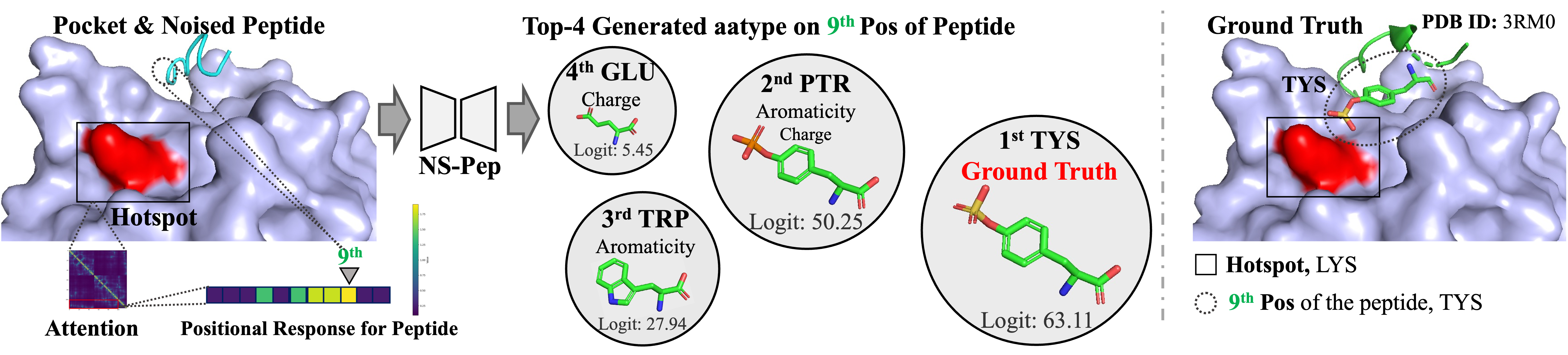}}
\caption{Interpretation of \ourmodel's generation for 3RM0 (PDB ID)}
\label{Fig.explain}
\vspace{-6mm}
\end{center}
\end{figure*}

\begin{figure*}[ht]
\vskip 0.2in
\begin{center}
\centerline{\includegraphics[width=\textwidth]{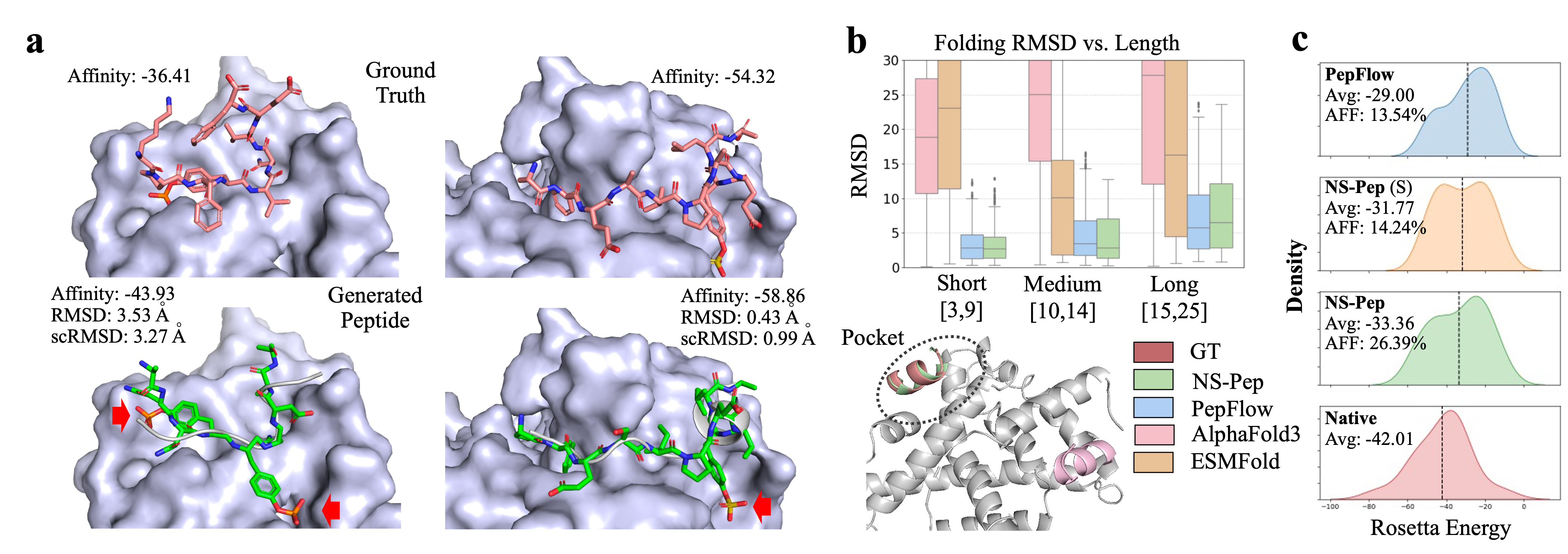}}
\caption{\textbf{Visualization of experimental results.} (\textbf{a}) Two generated peptides compared with their corresponding native ones (case: 1BMB, 3SHA). (\textbf{b}) \textbf{Top:} Folding performance across four models, stratified by peptide length. \textbf{Bottom:} Example showing a peptide predicted by AlphaFold3 that fails to bind the target pocket.  (\textbf{c}) Distribution of binding affinity. \ourmodel(S) substitutes all the NSAAs with their standard counterparts among the generation of \ourmodel.}
\label{Fig.exp}
\vspace{-6mm}
\end{center}
\end{figure*}

\subsection{Ablation Study} \label{Exp.Ablation}
\begin{table}[ht]
\centering
\caption{Ablation study for \ourmodel~on the \textbf{NSAA} test set.}
\resizebox{0.9\columnwidth}{!}{%
\begin{tabular}{lc|c|cccc|cccc}
\toprule
\multirow{4}{*}{\textbf{Methods}}
& NS & & \checkmark & \checkmark & \checkmark & \checkmark & \checkmark & \checkmark & \checkmark & \checkmark \\
& RFGM & & & \checkmark & & & & \checkmark & \checkmark & \checkmark \\ 
& PSP & & & & \checkmark & & \checkmark & & \checkmark & \checkmark \\
& IAW & & & & & \checkmark & \checkmark & \checkmark & & \checkmark \\ 
\midrule
\multirow{4}{*}{\textbf{Metrics}} 
& AAR \% $\uparrow$ & 30.80 & 30.01 & 30.72 & 30.46 & 29.62 & \underline{31.19} & 30.74 & 28.92 & \textbf{32.36} \\ 
& AAR(NS) \% $\uparrow$ & 0.00 & 2.08 & 19.05 & 13.10 & 12.80 & 10.12 & \underline{25.00} & 7.44 & \textbf{29.77} \\ 
& RMSD \AA $\downarrow$ & \textbf{3.04} & 4.20 & 3.60 & 5.28 & 5.54 & 5.69 & \underline{3.27} & 4.99 & 3.94 \\ 
& AFF \% $\uparrow$ & 13.54 & 13.19 & 26.04 & 24.65 & 15.63 & 14.58 & \textbf{27.78} & 10.42 & \underline{26.39} \\ 
\bottomrule
\end{tabular}%
}
\vspace{-4mm}
\label{tab.ablation}
\end{table}

This task is carried out on the \textbf{NSAA} test set to evaluate our approaches. Apart from the metrics mentioned before, we also investigate \textbf{AAR(NS)} that ONLY calculates the recovery of NSAAs. Note that \textbf{NS} denotes whether the model is extended to support NSAAs.

\textbf{Results.} Examining the AAR(NS) row in Tab.~\ref{tab.ablation}, we can draw the following conclusions: (i) each of the three additional methods individually improves this metric compared to the model only with NS; (ii) combining all these methods results in further enhancements; (iii) both the additive and subtractive experiments demonstrate the effectiveness of RFGM.

Fig.~\ref{Fig.exp}(\textbf{c}) validates the impact of NSAAs on affinity: substituting NSAAs in generated peptides (\ie \ourmodel~and \ourmodel(S)) leads to a sharp 12.15\% drop in AFF, highlighting their significant contribution.

\section{Conclusion and Limitation} \label{sec.conclusion}
In this study, we present \ourmodel, a peptide co-design framework capable of effectively incorporating non-standard amino acids (NSAAs). To address the intrinsic challenge posed by the long-tailed distribution of NSAAs, we introduce Residue Frequency-Guided Modification (RFGM), which calibrates output logits for rare residues. In addition, Progressive Side-chain Perception (PSP) enables accurate modeling of side-chain conformations, while Interaction-Aware Weighting (IAW) prioritizes residues proximal to the binding pocket.
\ourmodel~exhibits strong performance in \textit{de novo} peptide design, yielding improvements in both amino acid recovery and binding affinity, and demonstrates clear advantages in the pocket-specific peptide folding task. Nonetheless, its capacity to generalize to ultra-rare NSAA types (\eg, less than 10 samples) remains limited, suggesting the need for strategies such as few-shot learning.
Finally, further validation of designed peptides—particularly with respect to properties such as bioavailability and membrane permeability—remains an important direction for future work. Addressing these aspects may benefit from techniques such as reinforcement learning to iteratively refine the model for therapeutic applications. Note that \ourmodel~could be misused to design harmful peptides; users are urged to employ it responsibly. Overall, we expect that \ourmodel~will encourage further research into drug design with NSAAs.

{
\small
\bibliographystyle{IEEEtran}
\bibliography{MainTex}
}

\newpage
\appendix

\section{Gradient Analysis of Residue-wise Logit Modification} \label{sec.AnalysisRFGM}

For the standard used Cross-Entropy (CE) loss, we have:
\begin{equation*}
\mathcal{L}_{\text{CE}}(\mathbf{z}) = -\sum^{C}_{i=1}y_{i}\text{log}(p_{i}), \ \text{with} \ p_{i}=\frac{e^{z_{i}}}{\sum^{C}_{j=1}e^{z_{j}}} \in (0,1)
\end{equation*}
where $p_{i}$ is the $i$-th class probability calculated by softmax; ground truth label $y_{i} \in \{0,1\}$; $C$ represents the number of classes. Then, we can derive gradients of a training sample w.r.t. positive logit $z_i$ and negative logits $z_j$:
\begin{equation*}
\frac{\partial \mathcal{L}_{\text{CE}}(\textbf{z})}{\partial z_{i}} = p_{i} - 1 \ \ \in (-1,0)
\label{ori_pos_gra}
\end{equation*}

\begin{equation*}
\frac{\partial \mathcal{L}_{\text{CE}}(\textbf{z})}{\partial z_{j}} = p_{j} \ \ \in (0,1)
\label{ori_neg_gra}
\end{equation*}

According to the above equations, we can observe that the gradient is positively correlated to its class probability, \ie $\frac{\partial \mathcal{L}_{\text{CE}}(\textbf{z})}{\partial z_{k}} \propto p_{k}, \ k \in \{ i, j \}$. In other words, for any class $k$, we can analyze the trend of $\frac{\partial \mathcal{L}_{\text{CE}}(\textbf{z})}{\partial z_{k}}$ through $p_{k}$, regardless of whether it is a positive or negative label. Through optimization like gradient descent, the CE loss would increase the logit of the ground truth class and penalize that of fake labels in the meantime. We will analyze the mechanism of RFGM (decreasing the gradient penalty for tail classes) below.

\begin{lemma}
Let $k$ be any class. Suppose a non-negative perturbation $\epsilon \geq 0$ is \textbf{only} added to its corresponding logit, resulting in ${z}_{k}' = z_{k} + \epsilon$, while all other logits remain unchanged. Then, the gradient of the cross-entropy loss with respect to $z_k$ satisfies
\[
\frac{\partial \mathcal{L}_{\text{CE}}(\textbf{z}')}{\partial z_{k}} \geq \frac{\partial \mathcal{L}_{\text{CE}}(\textbf{z})}{\partial z_{k}},
\]
where $\textbf{z}'$ denotes the updated logit vector after perturbation.
\end{lemma}

\begin{proof}
First, we have $\frac{\partial \mathcal{L}_{\text{CE}}(\textbf{z})}{\partial z_{k}} \propto p_{k} = \frac{e^{z_{k}}}{\sum_{m \neq k} e^{z_{m}} + e^{z_{k}}}$ and $\frac{\partial \mathcal{L}_{\text{CE}}(\textbf{z}')}{\partial z_{k}} \propto p_{k}' = \frac{e^{(z_{k} + \epsilon)}}{\sum_{m \neq k} e^{z_{m}} + e^{(z_{k} + \epsilon)}}$. To prove $\frac{\partial \mathcal{L}_{\text{CE}}(\textbf{z}')}{\partial z_{k}} \geq \frac{\partial \mathcal{L}_{\text{CE}}(\textbf{z})}{\partial z_{k}}$, it is equal to prove $\frac{\partial \mathcal{L}_{\text{CE}}(\textbf{z}') /  \partial z_{k}}{\partial \mathcal{L}_{\text{CE}}(\textbf{z}) /  \partial z_{k}} - 1 \geq 0$:
\begin{align*}
\frac{\partial \mathcal{L}_{\text{CE}}(\textbf{z}') /  \partial z_{k}}{\partial \mathcal{L}_{\text{CE}}(\textbf{z}) /  \partial z_{k}} - 1 &= \frac{e^{\epsilon}(\sum_{m \neq k} e^{z_{m}} + e^{z_{k}}) - (\sum_{m \neq k} e^{z_{m}} + e^{(z_{k}+\epsilon)})}{\sum_{m \neq k} e^{z_{m}} + e^{(z_{k}+\epsilon)}} \\ &= \frac{(e^{\epsilon}-1)\sum_{m \neq k}e^{z_{m}}}{\sum_{m \neq k} e^{z_{m}} + e^{(z_{k}+\epsilon)}} \\ & \geq 0 \ (\because \epsilon \geq 0, \ \text{then} \  e^{\epsilon} \geq 1)
\end{align*}
\end{proof}

Note that when ${z}_{k}' = z_{k} - \epsilon, \ \epsilon \geq 0$, we can easily prove $\frac{\partial \mathcal{L}_{\text{CE}}(\textbf{z}')}{\partial z_{k}} \leq \frac{\partial \mathcal{L}_{\text{CE}}(\textbf{z})}{\partial z_{k}}$ by converting $\epsilon$ to $-\epsilon$ in the above proof. Now, we are about to demonstrate the gradient effect of RFGM, starting from two-class situation then generalizing to multi-class one.

\textbf{Two-class situation.} Assume that we have a more frequent and a less frequent class logits $z_L$ and $z_S$. Based on  the definition of RFGM in Eq.~\ref{eq.RFGM}, we add noise on them with different scales. For a simple yet accurate representation, these two modified logits are $\hat{z}_{k}=z_{k} + \lambda_{k} |\delta_{\sigma}^{k}|, i \in \{L, S\}$ and $\lambda_{L} \geq \lambda_{S}$. Recall that $\delta_{\sigma}^{k} \sim \mathcal{N}(0,\sigma)$ and  $\sigma$ is a hyperparameter which is identical for all classes. To prove the rare class suffers less gradient penalty after applied RFGM, it is equal to demonstrate $\frac{\partial \mathcal{L}_{\text{CE}}(\textbf{z}')}{\partial z_{S}} \leq \frac{\partial \mathcal{L}_{\text{CE}}(\textbf{z})}{\partial z_{S}}$. Because of the uncertainty introduced by noise $\delta_{\sigma}^{k}$, we can only prove it from the view of probability:

\begin{align*}
&P(\frac{\partial \mathcal{L}_{\text{CE}}(\textbf{z}')}{\partial z_{S}} \leq \frac{\partial \mathcal{L}_{\text{CE}}(\textbf{z})}{\partial z_{S}}) \\
&= P(p_{S}' \leq p_{S})\quad&(\frac{\partial\mathcal{L}_{\text{CE}}(\textbf{z})}{\partial z_{k}} \propto p_{k})\\
&=P(\frac{e^{(z_{S}+\lambda_{S}|\delta_{\sigma}^{S}|)}}{e^{(z_{S}+\lambda_{S}|\delta_{\sigma}^{S}|)}+e^{(z_{L}+\lambda_{L} |\delta_{\sigma}^{L}|)}} \leq \frac{e^{z_{S}}}{e^{z_{S}} + e^{z_{L}}})\\
&=P(\frac{e^{(z_{S} + (\lambda_{S} |\delta_{\sigma}^{S}|-\lambda_{L}|\delta_{\sigma}^{L}|))}}{e^{(z_{S} + (\lambda_{S} |\delta_{\sigma}^{S}|-\lambda_{L}|\delta_{\sigma}^{L}|))}+e^{z_{L}}} \leq \frac{e^{z_{L}}}{e^{z_{S}} + e^{z_{L}}})\\
&=P(\lambda_{S}|\delta_{\sigma}^{S}|-\lambda_{L}|\delta_{\sigma}^{L}| \leq 0)\quad&(\textbf{Lemma 1})\\
&=P(\frac{|\delta_{\sigma}^{S}|}{|\delta_{\sigma}^{L}|} \leq \frac{\lambda_{L}}{\lambda_{S}})\\
&=P(\frac{|V|}{|U|} \leq \frac{\lambda_{L}}{\lambda_{S}})\quad&(U=\frac{\delta_{\sigma}^{L}}{\sigma}\sim\mathcal{N}(0,1),V=\frac{\delta_{\sigma}^{S}}{\sigma}\sim\mathcal{N}(0, 1))\\
&=P(\text{arctan}(\frac{|V|}{|U|}) \leq \text{arctan}(\frac{\lambda_{L}}{\lambda_{S}})) \quad &(\text{arctan}(\cdot)\text{ is monotonically increasing})\\
&= \frac{2}{\pi}\text{arctan}(\frac{\lambda_{L}}{\lambda_{S}})\quad&\text{arctan}(\frac{|V|}{|U|})\sim U(0,\frac{\pi}{2})\text{\cite{Gaussian}}\\
& \geq \frac{1}{2} \quad&\frac{\lambda_{L}}{\lambda_{S}} \geq 1
\end{align*}

Thus, there is a greater than 0.5 probability that the gradient magnitude for the rare class decreases, indicating that the rare class experiences reduced gradient penalty due to the introduction of RFGM. It's also easy to prove $P(\frac{\partial \mathcal{L}_{\text{CE}}(\textbf{z}')}{\partial z_{L}} \geq \frac{\partial \mathcal{L}_{\text{CE}}(\textbf{z})}{\partial z_{L}}) \geq \frac{1}{2}$ according to the above derivation, which show that the gradient magnitude of frequent class is inclined to increase.

\textbf{Multi-class situation. }Generalizing to multi-class means introducing more perturbations to the above analysis, which is nontrivial to prove from probability viewpoint. From the point of expectation, we set $\gamma_{k}=\mathbb{E}[\lambda_{k}|\delta_{\sigma}^{k}|]$ and $\gamma_{1} 
\geq \cdots \geq \gamma_{N}$. For the most frequent class, we can approximate that:

\begin{equation*}
    p_{1}'\approx \frac{e^{z_1+\gamma_1}}{\sum(e^{z_k+\gamma_k})} \geq \frac{e^{z_1+\gamma_1}}{\sum e^{z_k+\gamma_1}}=\frac{e^{z_1}}{\sum e^{z_k}}=p_1
\end{equation*}

Recall that $\frac{\partial\mathcal{L}_{\text{CE}}(\textbf{z})}{\partial z_{k}} \propto p_{k}$, the modified gradient of the most frequent class, thus, is likely greater than its original gradient, \ie $\frac{\partial \mathcal{L}_{\text{CE}}(\textbf{z}')}{\partial z_{1}} \geq \frac{\partial \mathcal{L}_{\text{CE}}(\textbf{z})}{\partial z_{1}}$. Similarly, for the least frequent class, we have $p_{N}'\approx\frac{e^{z_N+\gamma_N}}{\sum(e^{z_k+\gamma_k})} \leq \frac{e^{z_N+\gamma_N}}{\sum e^{z_k+\gamma_N}}=p_N$. As for the a class $m \in \{2,\cdots,N-1\}$, its ratio $\frac{p_{m}'}{p_m}$ satisfies:

\begin{equation*}
    \frac{p_{m}'}{p_m} \approx \frac{\frac{e^{z_m + \gamma_m}}{\sum(e^{z_k + \gamma_k})}}{\frac{e^{z_m}}{\sum e^{z_k}}} = \frac{e^{\gamma_m} \cdot \sum e^{z_k}}{\sum(e^{z_k + \gamma_k})} \in [\frac{e^{\gamma_N} \cdot \sum e^{z_k}}{\sum(e^{z_k + \gamma_k})}, \frac{e^{\gamma_1} \cdot \sum e^{z_k}}{\sum(e^{z_k + \gamma_k})}] = [\frac{p_{N}'}{p_N}, \frac{p_{1}'}{p_1}]
\end{equation*} 

According to the above inequality, larger $\gamma_m$(s), from more frequent classes, lead to higher ratios $\frac{p_{m}'}{p_m}$ than those from less frequent classes. Thus, due to RFGM, the gradient magnitudes of head classes tend to rise, while those of tail classes incline to decrease.

Therefore, in terms of gradient descent, we can proved that \textit{RFGM enlarges the gradient penalty on the negative classes and reduces the gradient reward on the positive ones for the head class and conversely for the tail class}, as summarized in Tab.~\ref{tab.AnalysisRFGM}.

\begin{table}[ht]
\centering
\caption{The impact of RFGM on gradient}
\label{tab:AnalysisRFGM}
\begin{tabular}{lcccc}
\toprule
Class & Trend & True label? & Gradient & Description of back-propagation \\
\midrule
\multirow{2}{*}{Head} & \multirow{2}{*}{$p_{m}' \uparrow$}
& $\checkmark$ & $p_{m}' - 1$ & Gradient reward $\downarrow$ \\
& & $\times$ & $p_{m}'$ & Gradient penalty $\uparrow$ \\
\midrule
\multirow{2}{*}{Tail} & \multirow{2}{*}{$p_{m}' \downarrow$}
& $\checkmark$ & $p_{m}' - 1$ & Gradient reward $\uparrow$ \\
& & $\times$ & $p_{m}'$ & \textcolor{darkgreen}{Gradient penalty $\downarrow$} \\
\bottomrule
\end{tabular}
\label{tab.AnalysisRFGM}
\end{table}

\section{Experimental Details} \label{sec.AppendixExp}

\subsection{\ourmodel~ Implementation}
\label{sec.Implementation}

\textbf{Architecture.} 
The \textbf{encoders} of \ourmodel~consist of two distinct components: one dedicated to processing the pocket and the other to handling the perturbed complex. Specifically, the pocket encoder integrates residue types, all-atom coordinates, and backbone dihedral angles into a 128-channel single embedding. Additionally, it incorporates pairwise sequential identities, relative sequential positions, distogram, and relative orientations to generate a 64-channel pair embedding. The complex encoder, on the other hand, takes as input the pocket single embeddings, perturbed residue types, time embeddings, and angle embeddings, and outputs a 128-channel single embedding.
These embeddings are then fed into the denoiser module, as shown in Fig.~\ref{fig.MainFig}a. The \textbf{denoiser} comprises six recycles, each consisting of an Invariant Point Attention (IPA) transformer \cite{AF2}, a backbone update block, and an edge transition block. The IPA transformer leverages three tracks of information to update the single and pair embeddings while maintaining their channel numbers. Subsequently, the updated single embedding is used to refresh the backbone geometry and pair embedding. After completing the six recycles, the final single embedding is utilized to model other variables. Note that when predicting torsion angles and side-chain offsets, we only utilize a sub-network of \ourmodel~ with shared parameters but shallower layers. Because we assume that the prediction task is much simpler than the generation task. So the framework of \ourmodel~in Fig.~\ref{fig.MainFig}(\textbf{d}) is shadowed partially.

\textbf{Training.}
The training process of \ourmodel~adheres to Algorithm~\ref{alg.train}.

\ourmodel~ and all its variants in Sec.~\ref{Exp.Ablation} and Appendix.~\ref{Exp.comparison_longtail} are trained using Distributed Data Parallel mode on 2 NVIDIA A100 GPUs. Each model is trained for 320k steps with a learning rate of $1\times10^{-4}$. To monitor the performance of intermediate models, we generate approximately 200 peptides, randomly selected from the validation set, every 5k training step using 100 sampling timesteps. The validation loss is defined as $(\text{RMSD} - 10 \times \text{AAR})$. A plateau scheduler is employed with a decay factor of 0.8 and a patience of 10. We only choose the final checkpoint for the final test. The hyperparameter $\tau$ of IAW, \ie Eq.~\ref{eq.IAW}, is set to 4.5 \AA.

\begin{algorithm}[tb]
   \caption{Training process of \ourmodel}
   \label{alg.train}
\begin{algorithmic}[1]
   \STATE {\bfseries Input:} Dataset $\mathcal{D}$ with pre-computed weights (IAWs, Eq.~\ref{eq.IAW})
   \REPEAT
   \STATE Sample a complex $\{ \mathcal{I}^{\text{pep}},\mathcal{I}^{\text{poc}} \}$ from $\mathcal{D}$
   \STATE Perturb the peptide with a random time step, \\ $\mathcal{I}^{\text{pep}}_{t}=\{a^{j}_t, x^{j}_t, R^{j}_t \}^{n}_{j=1}, \ t \sim U(0,1)$
   \STATE Encode the fixed protein \\ $\mathbf{r}, \mathbf{p} = \mathbf{Encoder} (\mathcal{I}^{\text{poc}})$ \\ \# $\mathbf{r}$: single embedding, $\mathbf{p}$: pair embedding
   \STATE Predict the three modalities of peptide \\ $\{\hat{z}^{j}, \hat{x}^{j}, \hat{R}^{j} \}^{n}_{j=1} = \mathbf{Denoiser}(\mathbf{r}, \mathbf{p}, \mathcal{I}^{\text{pep}}_{t})$ 
   \IF{$t>0.75$}
        \STATE Predict the residue types and offsets \\ $\{\hat{\chi}^{j}, \Delta^j_{\text{SC}} \}^{n}_{j=1} = \mathbf{Denoiser}'(\mathbf{r}, \mathbf{p}, \{a^{j}, \hat{x}^{j}, \hat{R}^{j}\}^{n}_{j=1})$ \\ $a^j$: Ground-truth residue type
        \STATE Build the all-atom structure \\ $\{\hat{X}^j\}^{n}_{j=1} = \mathbf{BUILD}(\{a^j, \hat{x}^{j},\hat{R}^{j}, \hat{\chi}^{j}\}^{n}_{j=1})$ \\ $\{X^j\}^{n}_{j=1} = \mathbf{BUILD}(\{a^j, \hat{x}^{j},\hat{R}^{j}, \chi^{j}\}^{n}_{j=1})$
        \STATE Calculate the PSP loss (Eq.~\ref{eq.psp})
   \ENDIF
   \STATE Apply RFGM on logits $z^{j}$ (Eq.~\ref{eq.RFGM})
   \STATE Calculate the total loss (Eq.~\ref{eq.TotalLoss}) weighted by IAWs
   \STATE Update the model via backpropagation
   \UNTIL{converged}
\end{algorithmic} 
\end{algorithm}

\textbf{Generation.} All the generations are applied on an NVIDIA V100 GPU with 200 sampling timesteps. Note that \ourmodel~ generates $C_{\alpha}$ translations, backbone rotations and residue types first, then predicts the torsion angles and side-chain offsets based on the above variables. For each case of the test set, we generate 16 peptides for model evaluation.

\subsection{Data splitting}
\label{sec.DataSplit}

In this work, the benchmark dataset derives from Q-BioLip \cite{Q-BioLip} and PepBDB \cite{PepBDB}, comprising 10,348 protein-peptide complexes in total. We first extract receptor sequence(s) for every complex. If there are more than one chains, we concatenate them into one sequence for the following clustering. After that, we utilize mmseq2 \cite{mmseqs2} to cluster all the sequences, with a sequential similarity of 0.4. According to the clustering result, we split train: val: test set with ratio 8: 1: 1 as candidate pools. Then, noncontinuous peptides with large gaps ($>$ 3 a.a.) and over-similar complexes are removed for a high-quality test set. Specifically, we randomly select five samples from any cluster containing more than five elements. As a result, this procedure yields 8228:1065:219 complexes across the three datasets, with the test set referred to as the \textbf{General} test set. In order to investigate NSAAs and their impact on affinity in detail, we also define an \textbf{NSAA} test set from the \textbf{General} set, ensuring that each peptide contains at least one NSAA.
In the main manuscript we focus on the three most frequent NSAAs, each appears at least 200 times, and find that adding just these NSAAs already improves affinity substantially. In Appendix~\ref{sec.Generalization}, we further expand the set to 18 types, each appearing at least 30 times, and the \ourmodel~ still exhibits strong generalization ability.

Here are the details about sequential clustering for multi-chain receptor:
First, receptor sequences are extracted from all complex cases, with the number of receptor chains varying across cases. Next, all receptor chains are clustered using MMseqs2 with a sequence identity threshold of 0.4. Each case is then assigned cluster labels based on its receptor chains: single-chain cases receives one label, while multi-chain cases are assigned multiple labels. Cases are grouped into the same cluster only if their label sets are identical. For example, if case A has a single-chain receptor labeled (1), and cases B, C, D have two-chain receptors labeled (2, 2), (1, 2), and (2, 2), respectively, only cases B and D would belong to the same cluster. Notably, NSAAs in receptor sequences are replaced with \textit{UNK} before clustering. Since their occurrence is rare, this substitution has negligible impact on the clustering results.

\subsection{Baseline Details}
\label{sec.BaselineDetail}

\subsubsection{Protein (Peptide) Generation Models}

\textbf{RFdiffusion.} 
As a representative algorithm for \textit{de novo} protein design, RFdiffusion \cite{RFD} integrates a pretrained protein prediction model \cite{RF} into the diffusion process to generate backbone structures. In this study, we employ this approach to generate 16 peptide backbones for each pocket, using 200 sampling timesteps. Similar to \ourmodel, we provide only the protein pocket, rather than the full receptor, to RFdiffusion and designate all pocket residues as hotspot residues. Subsequently, ProteinMPNN \cite{ProteinMPNN} is applied to predict a sequence based on the given peptide backbone.

\textbf{ProteinGenerator.} 
To achieve consistency between sequence and backbone structure, ProteinGenerator \cite{ProteinGenerator}, also developed from RFdiffusion, performs diffusion in the sequence space guided by the structure. It then outputs proteins with compatible sequences and structures. Based on the official model parameters, we set all pocket residues as hotspot residues and generate 16 peptides per case to ensure fair comparison.

\textbf{PPFLOW.}
Given specific protein pockets, PPFLOW \cite{ppflow} generates peptide binders primarily by recovering their backbone dihedral angles through flow-matching. Following its official inference script and pretrained model, we utilize PPFLOW for peptide sequence-structure co-design with 500 sampling timesteps.

\textbf{DiffPepBuilder.}
DiffPepBuilder \cite{DPB}, a peptide sequence-structure co-design model, aims to generate potent and stable peptides targeting protein receptors. To achieve this, it incorporates co-evolution embeddings calculated by ESM2 \cite{ESM} and builds disulfide bonds. Notably, when generating peptides with this approach, we discard the disulfide-bond building function and set their length to match that of their native counterparts.

\textbf{PepGLAD.} Targeting a specific protein pocket, PepGLAD \cite{pepglad} first extracts structural features of the pocket, generates corresponding peptides in a learned latent space, and then decodes the resulting peptide–protein complex into all-atom 3D structures. We evaluate it on our test set using the official configuration provided by the authors.

\textbf{PepFlow.}
PepFlow \cite{pepflow} is a multi-modal all-atom peptide generation model based on flow-matching, achieving relatively high quality in peptide reconstruction. Before evaluation, we train this model from scratch with the same settings as our models, while retaining its original architecture.

\subsubsection{Long-tailed Learning Methods} \label{sec.app.longtail}

\textbf{Weighted Sampling.}
Weighted sampling is an intuitive approach that assigns higher probabilities to less frequent classes during training. Specifically, we first calculate the proportion of each amino acid in the current training set. Then, we define the weight of each type as the reciprocal of its proportion. Based on the residue-wise weights, we compute the sampling weight for each peptide, normalized by its sequence length. Thus, we can sample from the training set according to the weighted data distribution, rather than a uniform distribution.

By increasing the occurrence frequency of rare amino acids, the weighted sampling scheme helps generative models to better learn from them. However, it should be noted that this method may deviate the data distribution from the native one.

\textbf{Balancing Logit Variation.}
By introducing category-wise and frequency-related perturbations on classification logits during the training process, Balancing Logit Variation (BLV) modifies the area of each point in the feature space. As a result, the feature points of tail classes expand to larger areas, while those from head classes experience only smaller expansions. Specifically, the logit modification method of BLV is similar to our RFGM (Eq.~\ref{eq.RFGM}), but the former adjusts perturbations using counter-proportional weights:
\begin{equation}
\hat{z}_{i}^{BLV} = z_{i} + \frac{v_{i}}{max \ \mathbf{v}} |\delta(\sigma)|, \ with \ v_i = log \frac{\sum^{C}_{j=1}n_{j}}{n_{i}}
\label{eq.BLV}
\end{equation}

\textbf{Seesaw Loss.}
Initially applied in instance segmentation, Seesaw Loss \cite{seesaw} aims to address the issue that tail-class objects are more likely to be ignored or recognized as head categories. To achieve this, the work modifies the gradients of positive and negative samples using a mitigation factor and a compensation factor. Specifically, the mitigation factor reduces the gradient penalty on tail classes, preventing them from being excessively suppressed. Meanwhile, the compensation factor adjusts the gradients of false-positive samples to restrain them.

\subsection{Metric Details}
\label{sec.MetricDetail}
Generative models are developed to capture real data distribution and produce diverse samples. To be specific, given protein targets, our work aims to generate peptide binders with several desired attributes, \eg, high affinity, compatible sequence and structure and so on. The generated peptides are evaluated from three aspects. \cite{pepflow} 

(i) Recovery. Recovery metrics assess the similarity between real and generated peptides, w.r.t. sequence, structure, secondary structure and binding site. Although the purpose of a generative model is not to reconstruct the native samples merely, resembling them indicates that the generated samples are likely to possess their features, like considerable affinity and chemical validity. 
\begin{itemize}
    \item Amino Acid Recovery (\textbf{AAR}) calculates the ratio of the same sequential identity between the ground truth and designed peptides. To investigate the AAR for interested amino acid types, we also introduce  \textbf{AAR(NS)}, which is for the total recovery of the NSAAs.
    \item The \textbf{RMSD} can be seen as a structural recovery metric which aligns the complex by \textit{pockets} first and measures the RMSD between $C_{\alpha}$ atoms of peptide binders. This metric not only compares the structural similarity between peptides themselves but also considers the relative position between peptide and protein pockets.
    \item Secondary-Structure Recovery (\textbf{SSR}) denotes the ratio of shared secondary structure, calculated by the DSSP package.
    \item Binding-Site Recovery (\textbf{BSR}) quantifies the extent to which the generated binder can cover the binding site of its native counterpart. Specifically, the binding site consists of protein residues located within a 10 \AA \ distance from the peptide residues. It's mentioned that lower BSR leads to higher RMSD, because the peptide is far away from the native position.
\end{itemize}

(ii) Energy. The generated samples are not required to strictly resemble the ground truth. Instead, the goal is to design more potent binders that can attach tightly to their targets. To achieve this, we use Rosetta \cite{rosetta} to calculate the binding affinity energy.
\begin{itemize}
    \item The \textbf{AFF} metric calculates the percentage of generated samples whose affinity energy is lower than that of the native binder. Specifically, for all evaluations, we use the same truncated receptors as protein pockets and run the Rosetta energy function once.
\end{itemize}

(iii) Design. From the perspective of design, we want to generate chemically compatible and diverse peptides.
\begin{itemize}
    \item The chemical compatibility refers to the consistency between sequence and structure. Given the generated sequence and the sequence of the protein target, we use Chai-1 \cite{Chai-1} to predict the structure of the peptide-protein complex. Then, \textbf{scRMSD} refers to self-consistency RMSD that compares the difference between generated and predicted structures in terms of $C_{\alpha}$ atoms \cite{FrameDiff}. Notably, Chai-1 can directly process the three selected NSAAs and perform inference in a relatively short time compared to similar approaches \cite{af3}. Considerations about scRMSD: when Chai-1 predicts complex structure, it couldn't designate specific bind sites. Especially for large receptors with hundreds of residues, a mismatch of bind sites would lead to considerably high scRMSD.
    \item The \textbf{Div} metric quantifies the structural diversity of peptides by calculating the average value of one minus the pairwise TM-score \cite{tm-score} across the 16 generated samples. A higher \textbf{Div} score reflects greater divergence among the peptide structures.
\end{itemize}

\subsection{Additional Results}
\label{sec.AdditionalResult}

\subsubsection{Visualization of Additional generated peptides}
Additional peptides targeting different proteins are visualized in Fig.~\ref{Fig.App.Visualization_v1}. In the case of 1BA8, we observe that \ourmodel~ prefers to generate PTR at the last position, replacing the native TYS. Notably, PTR and TYS share similar physicochemical properties—namely, negative charge and high polarity—suggesting that \ourmodel~ effectively learns and leverages these properties from the training data. Similarly, in case 1R1P, sample 1, \ourmodel~ tends to substitute D (Aspartic Acid) at the second position with SEP, which also exhibits comparable characteristics, such as mild acidity and hydrophilicity. These observations further support the conclusion that \ourmodel~ is capable of capturing and generalizing physicochemical similarities among residues during generation.

\begin{figure*}[h]
\begin{center}
\centerline{\includegraphics[width=\textwidth]{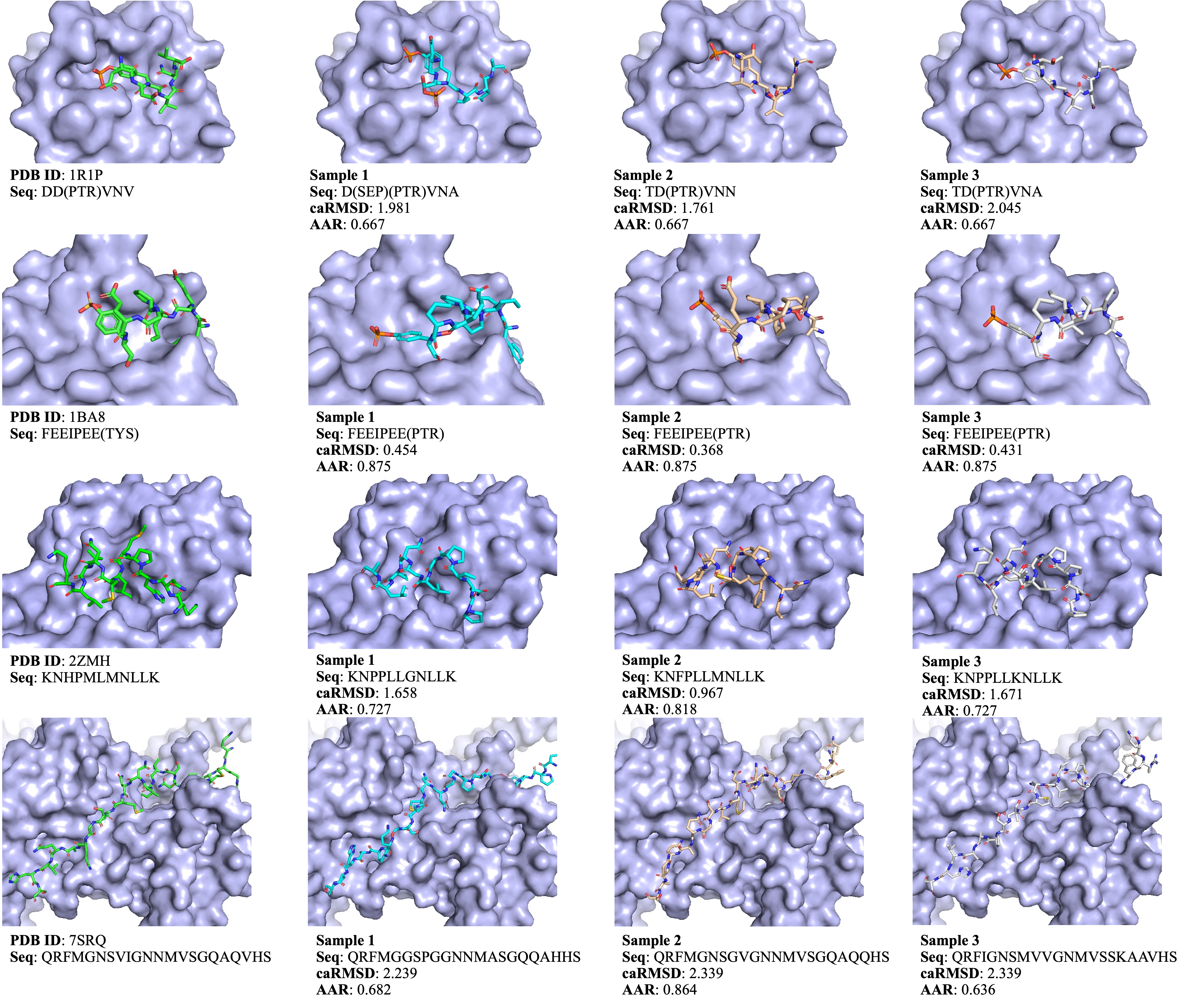}}
\caption{Additional peptides generated by \ourmodel}
\label{Fig.App.Visualization_v1}
\vspace{-6mm}
\end{center}
\end{figure*}

\subsubsection{Long-tailed Learning for NSAA Peptide Generation}
\label{Exp.comparison_longtail}

This task evaluates the effectiveness of various long-tailed learning approaches in capturing the NSAA distribution and investigates the impact of NSAAs on affinity. To emphasize the metrics influence of NSAAs, we conduct evaluations on the \textbf{NSAA} test set, which has a higher proportion of NSAAs compared to the \textbf{General} test set.

\begin{table}[ht]
\centering
\caption{Comparison experiments for long-tailed learning methods on the \textbf{NSAA} test set}
\resizebox{0.8\columnwidth}{!}{%
\begin{tabular}{lcccc}
\toprule
Method & AAR \% $\uparrow$ & AAR(NS) \% $\uparrow$ & RMSD \AA\ $\downarrow$ &  AFF \% $\uparrow$ \\
\midrule
PepFlow         & 30.80 & 0.00 & 3.04 & 13.54\\
PepFlow*    & 30.10 & 2.08 & 4.20 & 13.19\\
\midrule
PepFlow* + WS        & \underline{30.48} & \underline{4.76} & \underline{4.13} & 13.19\\
PepFlow* + BLV       & 28.93 & 2.08 & 4.40 & 15.63\\
PepFlow* + Seesaw    & 28.65 & 0.30 & 4.92 & \underline{16.32}\\
PepFlow* + RFGM       & \textbf{30.72} {\scriptsize  \textcolor{darkgreen}{\textbf{+0.24}}} & \textbf{19.05} {\scriptsize  \textcolor{darkgreen}{\textbf{+14.29}}} & \textbf{3.60} {\scriptsize  \textcolor{darkgreen}{\textbf{-0.53}}} & \textbf{26.04} {\scriptsize  \textcolor{darkgreen}{\textbf{+9.72}}} \\
\midrule
\ourmodel~ w/ None    & 31.19 & 10.12 & 5.69 & 12.28\\
\midrule
\ourmodel~ w/ WS    & 31.75 & 12.50 & 4.23 & 15.63\\
\ourmodel~ w/ BLV             & \textbf{32.39} & 25.00 & \textbf{2.88} & \underline{26.04}\\
\ourmodel~ w/ Seesaw          & 32.24 & \underline{27.08} & \underline{3.50} & 25.69\\
\ourmodel~ w/ RFGM      & \underline{32.36} {\scriptsize  \textcolor{darkred}{\textbf{-0.03}}} & \textbf{29.77} {\scriptsize  \textcolor{darkgreen}{\textbf{+2.69}}} & 3.94 {\scriptsize  \textcolor{darkred}{\textbf{+1.06}}} & \textbf{26.39} {\scriptsize  \textcolor{darkgreen}{\textbf{+0.35}}} \\
\bottomrule
\end{tabular}%
}
\vspace{-4mm}
\label{tab.comparison_longtail}
\end{table}

\textbf{Baselines.}
To prove the effectiveness of our long-tailed learning method, RFGM (Sec.~\ref{sec.RFGM}), we compare it with 
 (i) intuitive weighted sampling (WS), (ii) Balancing Logit Variation (BLV) \cite{BLV} and (iii) Seesaw loss \cite{seesaw} in the scenario of peptide generation with NSAAs. More details about these long-tailed learning approaches are presented in Appendix~\ref{sec.BaselineDetail}. Specifically, we use \textit{PepFlow} and \textit{PepFlow*} as our reference methods. The latter is a unified version of PepFlow capable of generating NSAAs. Then we compare these above long-tailed methods without and with side chain perception (Sec.~\ref{sec.SideChain}). For example, \textit{PepFlow* + WS} only utilizes the WS technique, while \textit{NS-Pep w/ WS} also adopts PSP and IAW. Moreover, \textit{NS-Pep w/ None} refers to PepFlow* equipped with side chain perception approaches, but without any long-tailed learning method.

\textbf{Metrics.} Focusing on the recovery of NSAAs and their impact on affinity, we use AAR,  \textbf{AAR(NS)}, RMSD, and AFF as evaluation metrics. In detail, AAR(NS) only assesses the recovery of NSAAs. Other metrics follow the same definitions in Sec.~\ref{Exp.Comparison}.

\textbf{Results.} As shown in Table~\ref{tab.comparison_longtail}, PepFlow* struggles to recover the NSAAs without long-tailed learning methods, only with 2.08\% recovery rate. The model achieves a higher AAR(NS) when equipped with WS or RFGM. Notably, RFGM alone improves this metric by 16.97\%, highlighting its effectiveness in residue-wise long-tailed learning. However, incorporating BLV or Seesaw does not yield any improvement in this regard.
Introducing side-chain perception techniques enhances all models in terms of both AAR and AAR(NS). Among these methods, RFGM remains the most effective for NSAA recovery and AFF, even when combined with side-chain techniques, underscoring its robustness. While Seesaw loss alone struggles with long-tailed distributions, it becomes highly effective when paired with side-chain methods, achieving an AAR(NS) of 27.08\%.
Comparing these performance trends, we observe that both a high AAR(NS) and a low RMSD contribute to a high AFF. Thus, although \ourmodel~ w/ None outperforms PepFlow* in NSAA recovery by a significant margin, it exhibits a lower AFF value, likely due to its weaker structural learning ability (\ie, higher RMSD).
Furthermore, while \ourmodel~ w/ RFGM falls behind BLV in both AAR and RMSD, it achieves a higher AFF, likely due to its superior NSAA learning capability combined with adequate structural recovery.

\subsubsection{Impact of Noise on the RFGM and BLV}
\label{sec.Ablation_noise}
RFGM and BLV mainly differs from the scales of noise. To investigate the impact of different components in RFGM and BLV, particularly the role of scaling, we conduct an ablation study on the General test set. Since both RFGM and BLV introduce noise to the predicted logits, we hypothesize that this noise—akin to a noising-denoising mechanism in self-supervised learning—enhances sequential learning. We analyze the effects by selectively removing the noise components and present the results in Table~\ref{tab.AblationNoise}. For better visualization, three key metrics are also plotted in Figure~\ref{Fig.ablation_noise}.

\begin{table}[ht]
\centering
\caption{Ablation study for noise of RFGM and BLV on the \textbf{General} test set}
\label{tab.AblationNoise}
\resizebox{\textwidth}{!}{%
\begin{tabular}{lcccccc}
\toprule
Method & AAR \% $\uparrow$ & AAR(SEP) \% $\uparrow$ & AAR(TYS) \% $\uparrow$ & AAR(PTR) \% $\uparrow$ & AAR(NS) \% $\uparrow$ & RMSD (\AA) $\downarrow$ \\
\midrule
\textbf{\ourmodel (RFGM)} & 22.53 & 0.57 & 64.29 & 56.25 & 29.77 & 4.92 \\
\ourmodel (RFGM w/o noise) & 15.81 & 7.39 & 41.07 & 64.58 & 26.79 & 5.82 \\
\ourmodel (BLV) & 20.99 & 1.14 & 61.61 & 27.08 & 25.00 & 5.22 \\
\ourmodel (BLV w/o noise) & 17.49 & 0.00 & 8.04 & 10.42 & 4.17 & 5.08 \\
\bottomrule
\end{tabular}%
}
\end{table}

\begin{figure*}[hbt]
\begin{center}
\centerline{\includegraphics[width=0.45\textwidth]{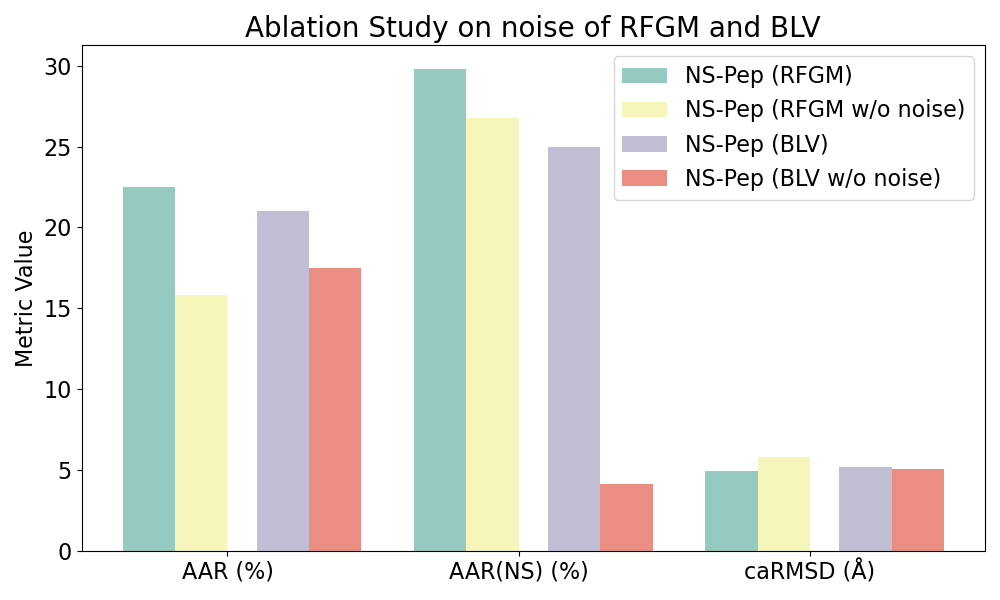}}
\caption{Ablation study on noise of RFGM and BLV}
\label{Fig.ablation_noise}
\vspace{-6mm}
\end{center}
\end{figure*}

The results clearly show that removing the noise leads to a notable drop in \textbf{AAR}, with both RFGM and BLV experiencing at least a 3.5\% decrease. This highlights the critical role of noise in promoting sequential predictions. Once the shared perturbation component is removed, only the scaling mechanisms in RFGM and BLV remain active. We further observe a performance drop in \textbf{AAR(NS)}, particularly for BLV, where the degradation is substantial. In contrast, RFGM shows minimal performance loss, demonstrating the effectiveness of its scaling strategy. Interestingly, despite these changes, \textbf{RMSD} remains largely unaffected, suggesting that structural accuracy is less sensitive to these modifications.

\subsubsection{Additional Interpretability Experiment} \label{sec.App.Explain}
To illustrate how \ourmodel~ determines the placement of NSAAs, we take 1BMB (PDB ID) as an additional example, focusing on residue 4 of the peptide chain I. This residue, PTR, extends deep into the binding pocket and interacts with residue 86 (ARG, R) in chain A. At this site, \ourmodel~ achieves a 75\% amino acid recovery (AAR).
(i) During inference, the attention matrix from the IPA block reveals that residue 86 pays the most attention to residues 4 and 5 of the peptide, indicating that \ourmodel~ effectively identifies key interaction sites.
(ii) In the fifth generated peptide, although PTR is not ranked highest, \ourmodel~ instead assigns high probabilities to structurally and chemically related amino acids. The top-5 predictions are GLU (E), LYS (K), ARG (R), TYR (Y), and PTR. Among them, GLU shares similar polarity and hydrophilicity with PTR; LYS and ARG possess long side chains that can reach into the pocket like PTR; and TYR is the canonical precursor of PTR. These predictions suggest that \ourmodel~ captures both the geometric context and relevant physicochemical properties of the binding site. Additional evidence is provided by the first two cases in Fig.~\ref{Fig.App.Visualization_v1}, which further support this observation.

\begin{figure*}[hbt]
\begin{center}
\centerline{\includegraphics[width=\textwidth]{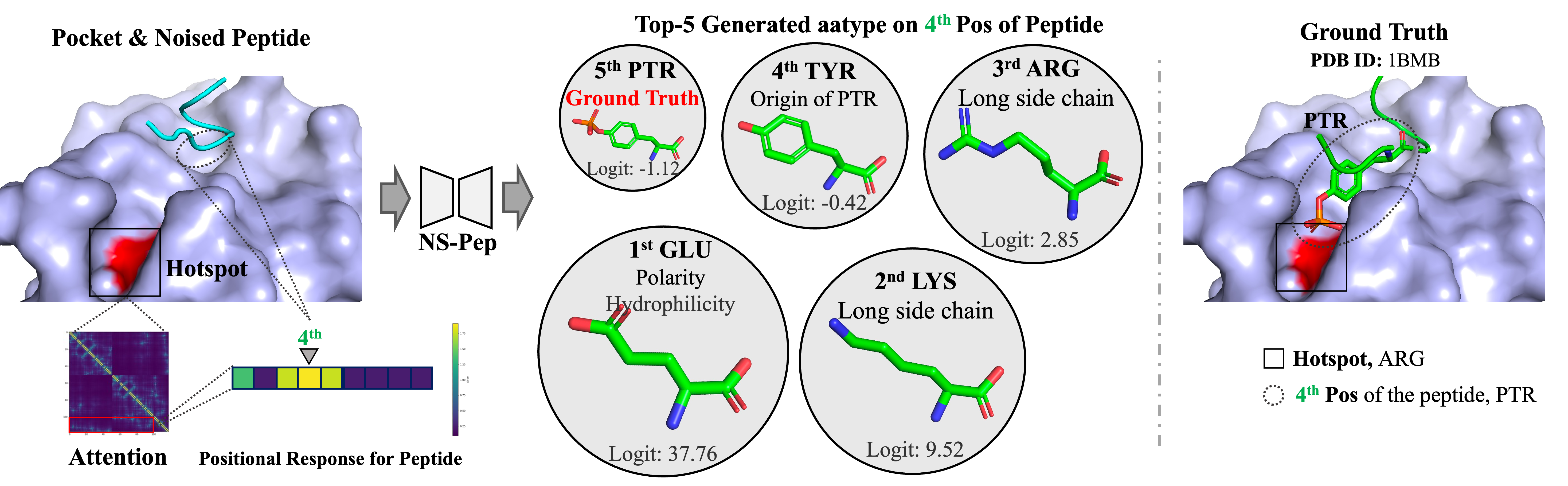}}
\caption{Interpretation of \ourmodel~ generation for 1BMB (PDB ID)}
\label{Fig.explain_app}
\vspace{-6mm}
\end{center}
\end{figure*}

\subsubsection{Hyperparameter Choice Experiments} \label{sec.app.hyperparam}
In this section, we conduct experiments to explore the choice of hyperparameters (Tab.~\ref{tab.hyperparam}), including our \ourmodel~ model and other long-tailed learning approaches. The evaluation metrics are defined identically to those in Sec.~\ref{sec.exp}. 

\textbf{Hyperparameter explanation.} The experiments for selecting hyperparameters involve four types of hyperparameters: (i) \textit{SC} refers to the weight $\lambda_{\text{SC}}$ for our PSP loss (Eq.~\ref{eq.psp}), serving to regulate the intensity of supervision derived from side-chain geometry; (ii) \textit{RFGM} indicates the standard deviation $\sigma$ in Eq.~\ref{eq.RFGM} that indirectly modulates the magnitude of gradients for both head and tail classes; (iii) \textit{BLV} shares a similar interpretation as \textit{RFGM} but pertaining specifically to Eq.\ref{eq.BLV}; (iv) \textit{Seesaw} comprises two weights for the mitigation factor and compensation factor, respectively.

\textbf{Results.} 
Obviously, we see that PepFlow* struggles to improve AAR(NS) merely by tuning the BLV or Seesaw hyperparameters; neither method substantially addresses the residue‐wise long‐tailed distribution.

Turning to NS‐Pep with different hyperparameters, increasing the RFGM parameter $\sigma$ often leads to higher AAR, and when $\lambda_{\text{SC}}$ is kept in $\{ 0.01, 0.05, 0.1 \}$, also improves AFF. On the flip side, a larger $\lambda_{\text{SC}}$ tends to worsen RMSD but increase design diversity—a clear trade‐off between backbone fidelity and exploration. Notably, $\lambda_{\text{SC}}=0.1$ and $\sigma=15$ seem to strike the best balance among AAR, AAR(NS), and AFF.

Although the BLV method does not remedy the long‐tail issue (as shown by relatively flat and low AAR(NS)), it does boost AAR, indicating its potential for classification‐oriented tasks. Moreover, a larger $\sigma$ in BLV appears to raise SSR, suggesting some benefit for secondary structure recovery.

Finally, with the Seesaw loss, a higher mitigation factor reduces AAR, while increasing the compensation factor can recover or improve AAR but at the cost of a higher RMSD. Thus, each hyperparameter choice reflects a nuanced compromise among structure recovery, sequence recovery, and design diversity.

\begin{table*}[htb]
\centering
\caption{Hyperparameter choice experiments on the \textbf{General} test set}
\resizebox{\textwidth}{!}{%
\begin{tabular}{lccccccc}
\toprule
\multirow{2}{*}{Method} & \multicolumn{5}{c}{Recovery} & \multicolumn{1}{c}{Energy} & \multicolumn{1}{c}{Design} \\
\cmidrule(r){2-6} \cmidrule(r){7-7} \cmidrule(r){8-8}
 & AAR \% $\uparrow$ & AAR(NS) \% $\uparrow$  & RMSD \AA\ $\downarrow$ & SSR \% $\uparrow$ & BSR \% $\uparrow$ & AFF \% $\uparrow$ & Diversity $\uparrow$ \\
\midrule
PepFlow & 16.30 & 0.00 & 4.98 & 85.80 & 81.50 & 11.60 & 0.50 \\
PepFlow* & 15.90 & 2.08 & 4.94 & 83.40 & 84.36 & 11.39 & 0.51 \\
PepFlow* + RFGM & 20.51 & 19.05 & 4.99 & 86.05 & 81.78 & 13.49 & 0.50 \\
\midrule
\ourmodel(SC 0.01, RFGM 10) & 21.99 & 7.74 & 5.41 & 85.37 & 81.55 & 14.45 & 0.50 \\
\ourmodel(SC 0.01, RFGM 15) & 20.96 & 13.69 & 4.99 & 84.92 & 80.68 & 14.58 & 0.51 \\
\ourmodel(SC 0.01, RFGM 20) & \underline{22.18} & 27.38 & 5.35 & 83.64 & 81.77 & \underline{15.25} & 0.52 \\
\ourmodel(SC 0.05, RFGM 10) & 20.14 & 14.88 & \textbf{4.79} & 86.00 & 81.93 & 13.64 & 0.48 \\
\ourmodel(SC 0.05, RFGM 15) & 21.00 & 2.98 & 5.39 & 86.10 & \underline{82.27} & 13.74 & 0.51 \\
\ourmodel(SC 0.05, RFGM 20) & 20.95 & 16.07 & 5.20 & 85.50 & 81.08 & 14.48 & \textbf{0.53} \\
\ourmodel(SC 0.1, RFGM 10) & 20.27 & 28.57 & 5.49 & 85.63 & \textbf{82.31} & 13.64 & 0.50 \\
\textbf{\ourmodel(SC 0.1, RFGM 15)} & \textbf{22.53} & \textbf{29.77} & \underline{4.92} & 85.00 & 81.62 & \textbf{17.23} & 0.49 \\
\ourmodel(SC 0.1, RFGM 20) & 21.79 & 18.15 & 5.37 & \underline{86.41} & 82.16 & 14.23 & 0.52 \\
\ourmodel(SC 0.15, RFGM 10) & 20.13 & 23.51 & 6.17 & 84.83 & 81.83 & 13.89 & 0.52 \\
\ourmodel(SC 0.15, RFGM 15) & 17.96 & 6.85 & 6.23 & \textbf{86.45} & 81.30 & 13.40 & 0.51 \\
\ourmodel(SC 0.15, RFGM 20) & 19.73 & 25.30 & 5.74 & 85.71 & 79.77 & 12.81 & \textbf{0.53} \\
\ourmodel(SC 0.2, RFGM 20) & 20.85 & \textbf{29.77} & 6.11 & 82.73 & 80.76 & 12.78 & \textbf{0.53} \\
\midrule
PepFlow* + BLV 5 & 17.6 & 0.89 & \underline{4.95} & 85.07 & 82.71 & 13.87 & \underline{0.48} \\
PepFlow* + BLV 10 & \underline{21.2} & 0.89 & 5.19 & 85.90 & \textbf{84.30} & 12.41 & \underline{0.48} \\
PepFlow* + BLV 15 & \textbf{21.9} & \underline{2.08} & \textbf{4.86} & \underline{85.96} & \underline{84.13} & \textbf{16.31} & 0.47 \\
PepFlow* + BLV 20 & 20.56 & \textbf{2.38} & 5.15 & \textbf{87.04} & 82.67 & \underline{13.90} & \textbf{0.51} \\
\midrule
PepFlow* + Seesaw(0.2, 1.0) & \underline{15.96} & \underline{2.38} & \textbf{4.79} & \textbf{86.12} & 83.44 & \textbf{13.99} & \underline{0.48} \\
PepFlow* + Seesaw(0.2, 2.0) & \textbf{16.01} & \underline{2.38} & 5.33 & 83.79 & \textbf{84.00} & 13.46 & 0.46 \\
PepFlow* + Seesaw(0.8, 1.0) & 14.48 & \textbf{2.68} & \underline{5.00} & 85.00 & 83.22 & 9.62 & \textbf{0.49} \\
PepFlow* + Seesaw(0.8, 2.0) & 15.10 & 0.30 & 5.09 & \underline{85.97} & \underline{83.85} & \underline{13.96} & \underline{0.48} \\

\bottomrule
\end{tabular}%
}
\vspace{-2mm}
\label{tab.hyperparam}
\end{table*}

\subsubsection{Generalization ability of NS-Pep
} \label{sec.Generalization}

To further prove the scalability of NS-Pep, we include additional NSAAs that appear more than 30 times in the dataset, resulting in a total of 38(=20+18) amino acid types. They are: SEP (O-Phosphoserine), TYS (O-Sulfotyrosine), PTR (O-Phosphotyrosine), MLE (N-Methyl-L-glutamic acid), M3L (N-Trimethyl-L-lysine), ALY (N6-Acetyl-L-lysine), TPO (O-Phosphothreonine), DAL (D-Alanine), HYP (4-Hydroxyproline), MVA (N-Methylvaline), DLY (D-Lysine), DLE(D-Leucine
), SAR (Sarcosine), NLE (Norleucine), BMT ($\beta$-Methylthiotyrosine), DGL (D-Glutamic acid), DPR (D-Proline), DTR (D-Tryptophan). Since NSAAs in the validation set are not used for evaluation, we reassign the cases comprising them to the training and test sets while keeping the dataset sizes unchanged, i.e., 8228:1065:219. Apart from this adjustment in data usage, all other experimental settings remain the same. Tab.~\ref{tab.NS18_comparison} and Fig.~\ref{Fig.4.AAR_NS18} summarize the performance of PepFlow* and NS-Pep under this setting. This comparison clearly shows that \ourmodel~ outperforms PepFlow* in recovering all NSAAs, highlighting its effectiveness in long-tailed learning.

\begin{table}[ht]
\centering
\caption{Comparison of PepFlow* and \ourmodel~ on 18 types of NSAA}
\resizebox{1.0\columnwidth}{!}{%
\begin{tabular}{lcccccccc}
\toprule
Method & AAR \% $\uparrow$ & RMSD (\AA) $\downarrow$ & Diversity $\uparrow$ & AAR(NS) \% $\uparrow$ & AAR(SEP) & AAR(TYS) & AAR(PTR) & AAR(MLE) \\
\midrule
PepFlow* & 15.45 & 4.90 & 0.49 & 1.43 & 1.74 & 0.00 & 0.00 & 0.00 \\
\ourmodel  & 25.02 & 4.79 & 0.56 & 29.36 & 44.44 & 17.71 & 59.66 & 25.42 \\
\midrule
\multicolumn{9}{c}{\textit{Continued: AAR for other NSAAs}} \\
\midrule
Method & AAR(M3L) & AAR(ALY) & AAR(DAL) & AAR(HYP) & AAR(MVA) & AAR(DLY) & AAR(DLE) & AAR(SAR) \\
\midrule
PepFlow* & 0.00 & 0.57 & 11.54 & 0.00 & 3.13 & 0.00 & 0.00 & 1.56 \\
\ourmodel  & 8.59 & 32.10 & 14.90 & 0.00 & 37.50 & 63.02 & 13.75 & 45.31 \\
\midrule
\multicolumn{9}{c}{\textit{Continued: AAR for remaining NSAAs}} \\
\midrule
Method & AAR(NLE) & AAR(BMT) & AAR(DGL) & AAR(DPR) & AAR(DTR) & AAR(TPO) & - & - \\
\midrule
PepFlow* & 0.00 & 0.00 & 0.00 & 0.00 & 0.00 & 0.00 & - & - \\
\ourmodel  & 0.00 & 28.13 & 0.00 & 0.00 & 6.25 & 0.00 & - & - \\
\bottomrule
\end{tabular}
}
\label{tab.NS18_comparison}
\end{table}

\begin{figure*}[ht]
\begin{center}
\centerline{\includegraphics[width=\textwidth]{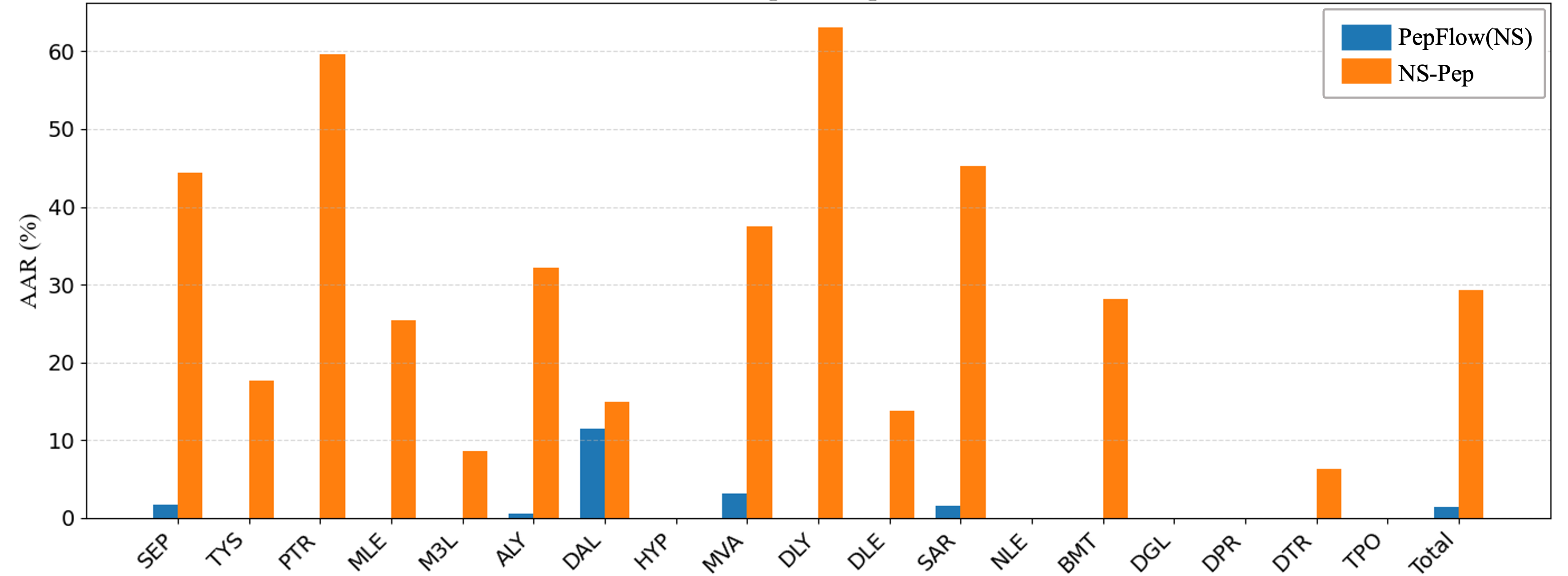}}
\caption{Amino acid recovery of \ourmodel~ and PepFlow on 18 types of NSAA}
\label{Fig.4.AAR_NS18}
\vspace{-8mm}
\end{center}
\end{figure*}

\section{Statistics of Side-chain Torsion Angle}
In our dataset, we perform a statistical analysis of side-chain torsion angles across 38 amino acid types and visualize their distributions in Fig.~\ref{Fig.torsion_dist}. For ALY, a non-standard amino acid with six rotatable side-chain bonds, we extend the number of considered $\chi$ angles to six. Additionally, to account for the CN atom (an N-terminal modification present in MLE, MVA, SAR, and BMT), we include its torsion angle in the $\chi$-angle system.

The resulting distributions reveal that it is difficult to accurately infer the identity of a residue based solely on its torsion angle(s). For instance, PHE, TYR, and TRP all possess two rotatable bonds and exhibit highly similar torsion angle distributions. While the similarity between PHE and TYR is chemically reasonable—differing mainly by a terminal oxygen atom—TRP also displays comparable torsion profiles despite having a markedly distinct side-chain structure. If a model is misled by these analogous distributions, it may generate incorrect residue predictions, leading to significant errors in physicochemical interpretation.

Moreover, we observe that most torsion angles tend to cluster within specific ranges, particularly for BMT. This suggests that, once a residue type is determined, predicting its torsion angles becomes relatively straightforward.

\begin{figure*}[h]
\begin{center}
\centerline{\includegraphics[width=\textwidth]{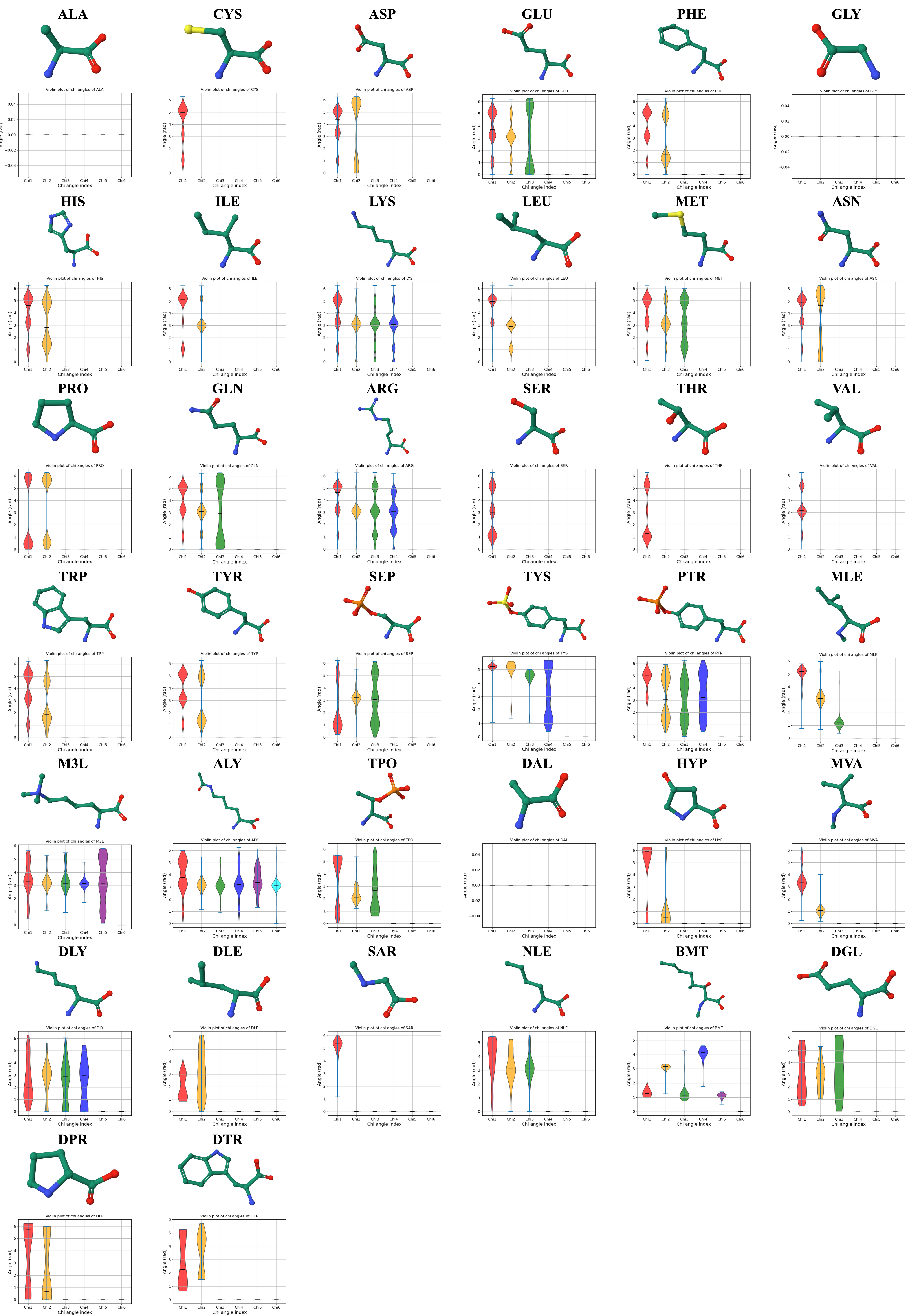}}
\caption{Distributions of torsion angles for 38 amino acid types}
\label{Fig.torsion_dist}
\end{center}
\end{figure*}


\end{document}